%% file: sbl_tsp_technical_main.tex
\documentclass[journal]{IEEEtran}
\input{sbl_tsp_setup.tex}

\begin{document}
\title{Large-System Analysis of Sparse Bayesian Learning}
\author{
	Fangqing~Xiao,~\IEEEmembership{Member,~IEEE}, \quad
	Dirk~T.~M.~Slock,~\IEEEmembership{Life Fellow,~IEEE}
	
	\thanks{
		Fangqing Xiao is with the School of Information Science and Engineering, 
		Yunnan University, Kunming 650091, China (e-mail: fangqing.xiao@ynu.edu.cn).
	}
	
	\thanks{
		Dirk T. M. Slock is with the Communication Systems Department, 
		EURECOM, 06410 Biot, France (e-mail: dirk.slock@eurecom.fr).
	}
	
	\thanks{
		Corresponding author: Fangqing Xiao (e-mail: fangqing.xiao@ynu.edu.cn).
	}}
\maketitle
\input{sbl_tsp_front.tex}
\input{sbl_tsp_body.tex}
\input{sbl_tsp_appendices.tex}

\input{sbl_tsp_refs.tex}
\end{document}

%% file: sbl_tsp_setup.tex
\usepackage{amsmath,amssymb}
\usepackage{bm}
\usepackage{cite}

\newtheorem{theorem}{Theorem}
\newtheorem{proposition}{Proposition}
\newtheorem{corollary}{Corollary}
\newtheorem{lemma}{Lemma}
\newtheorem{definition}{Definition}
\newtheorem{assumption}{Assumption}

%% file: sbl_tsp_front.tex
\begin{abstract}
Sparse Bayesian learning is widely used for sparse linear inverse problems, yet
its large-system stationary behavior remains poorly understood because all
variance hyperparameters are estimated from the same data. We study classical
sparse Bayesian learning, formulated as evidence maximization (type-II maximum
likelihood), for underdetermined linear models with sensing matrices having independent and
identically distributed Gaussian entries, Gaussian measurement noise, and an
unknown deterministic signal sequence. Hyperparameter reoptimization induces a nonvanishing
feedback term: a typical coordinate obeys a reoptimization-corrected scalar
Gaussian law whose signal coefficient is governed by the normalized adaptive
response rather than the frozen resolvent trace. A one-coordinate leave-one-out
construction gives an exact conditional Gaussian law, which is transferred to
a selected full stationary branch without assuming asymptotic closeness of the
reduced and full stationary vectors. Combining this law with the
Karush--Kuhn--Tucker conditions of the evidence objective yields a generally
set-valued scalar relation and three branchwise large-system consistency
relations. If the model noise variance is jointly estimated by evidence
maximization, interior joint stationarity yields an exact finite-dimensional
equality between the normalized residual energy and normalized resolvent trace.
When the limiting signal law has nonzero mass at zero, this identity further
yields a parameter-free asymptotic chi-square null law. Under an additional differentiability condition
on the selected scalar branch, the large-system characterization also gives a
closed relation for the reconstruction error of the posterior mean. The
analysis is stationary-point based and permits multiple stationary branches.
\end{abstract}


\section{Introduction}

Sparse Bayesian learning (SBL) is a widely used framework for sparse
linear inverse problems~\cite{Tipping2001,TippingFaul2003,WipfRao2004} and has
been applied to array processing, channel estimation, source imaging, and other
sensing problems~\cite{Dai2015,Srivastava2019,Ojeda2018,Grebien2024}. In classical SBL, the coefficients have independent zero-mean Gaussian priors with
variance hyperparameters, which are estimated by maximizing the marginal likelihood
(evidence), equivalently by type-II maximum likelihood within an empirical Bayes framework
\cite{Tipping2001,WipfRao2004}. We impose no hyperprior on these variances; the model
noise variance is fixed or estimated from the same evidence objective.

Existing SBL theory has studied the geometry of the evidence objective, sparsity-inducing
stationary points, reweighted and latent-variable interpretations,
hyperparameter optimization, and convergence properties
\cite{FaulTipping2001,WipfRao2004,WipfNagarajan2007,WipfNagarajan2010,
	WipfRaoNagarajan2011,Hashemi2021,Yu2024,Song2024}. Related developments in empirical Bayes sparse recovery include Bayesian compressive sensing and
simultaneous or block-structured SBL formulations
\cite{Ji2008,WipfRao2007,ZhangRao2013}. Complementary results address estimation limits,
support recovery, and SBL pruning mechanisms
\cite{PrasadMurthy2013,KhannaMurthy2022,Moderl2025}. Off-grid and multisnapshot
array formulations further illustrate SBL evidence maximization in
structured sensing models \cite{YangXieZhang2013,Gerstoft2016}.
Recent empirical Bayes theory has also studied automatic relevance determination and
consistency for lasso-type regularizers under a different nonconjugate model
\cite{YoshidaWatanabe2026}. Our preliminary conference study derived an SBL fixed-point
law under Gaussian designs \cite{XiaoSlock2026ICASSP}. Here we analyze fully reoptimized
stationary branches through a genuine leave-one-out (cavity) construction and the feedback induced by hyperparameter
reoptimization, without assuming full-to-reduced stationary-vector proximity. To our knowledge,
this reoptimization-corrected coordinate law has not been established in the proportional large-system regime.

State evolution (SE) gives scalar characterizations for approximate message passing (AMP), vector approximate message passing (VAMP), and related iterative estimators
\cite{Donoho2009,BayatiMontanari2011,Rangan2019,VilaSchniter2013}, while rigorous
leave-one-out representations are available for AMP under Gaussian designs \cite{BaoHanXu2025}.
Exact error asymptotics are also known for several convex or explicitly regularized estimators
\cite{BayatiMontanari2012,Thrampoulidis2015,Thrampoulidis2018}. These analyses concern a
prescribed algorithm or optimizer. We instead characterize a selected branch of Karush--Kuhn--Tucker (KKT) points of the
nonconvex SBL evidence objective, allowing multiple stationary points and limiting scalar roots.

The central difficulty is dependence created by hyperparameter reoptimization: deleting a sensing
column from a full-data resolvent does not make the remaining hyperparameters independent of that
column. We instead delete both the target column and its true signal contribution, producing a reduced
SBL problem independent of the deleted Gaussian column. Local continuations return to the selected full
stationary branch without assuming that the genuine reduced and full stationary vectors are asymptotically close. This construction reveals nonvanishing feedback induced by reoptimizing the
remaining hyperparameters. The resulting scalar Gaussian law is governed by the normalized adaptive
response rather than the frozen resolvent trace. Combining it with the finite-dimensional KKT conditions
yields a branchwise scalar KKT relation and large-system consistency relations. Although the resulting scalar characterization is compact, its derivation is not a direct consequence of Gaussian sensing. The proof must retain the nonvanishing response of the reoptimized hyperparameters, control piecewise-smooth support changes and a random continuation endpoint, handle the dependence of the adaptive residual on the observation noise, and use a two-coordinate cavity to obtain empirical self-averaging.

Our contributions are threefold. First, we derive a reoptimization-corrected scalar Gaussian law for a
typical coordinate of a regular stationary branch of the SBL evidence objective, without assuming asymptotic
closeness of the genuine reduced and full stationary vectors or a fixed active hyperparameter set.
Second, the finite-dimensional KKT conditions yield a generally set-valued scalar KKT relation and
three branchwise large-system consistency relations, allowing multiple stationary points and scalar
roots. Third, joint estimation of the model noise variance by evidence maximization gives the exact identity
$V_n=\eta_n$ at an interior joint stationary point and, when $P_X(\{0\})>0$, a parameter-free
asymptotic chi-square null law. An additional scalar-branch differentiability condition yields a
closed relation for the mean-square error (MSE) of the posterior mean. No generative prior or sparsity assumption is imposed on the
deterministic true signal sequence. The resulting scalar equations are necessary branchwise consistency conditions for realized regular subsequential limits; they are not an iterative update, and we do not claim that every algebraic scalar root is realizable by SBL.

\emph{Notation:}
Bold lowercase and uppercase symbols denote vectors and matrices,
respectively; bold Greek symbols are used analogously. Unbold italic symbols
denote scalars, including vector components and matrix entries, unless stated
otherwise (e.g., $x_i$ and $A_{\ell i}$). The $i$th column of $\mathbf A$ is
$\mathbf a_i$, and $(\cdot)^{\mathsf T}$ denotes transpose. For vectors,
$\|\cdot\|=\|\cdot\|_2$ and $\|\cdot\|_\infty$ denote the Euclidean and maximum
norms; for matrices, $\operatorname{tr}(\cdot)$, $\|\cdot\|_{\mathrm{op}}$, and
$\|\cdot\|_{\mathrm F}$ denote trace, operator norm, and Frobenius norm.
$\mathbf I_m$ is the $m\times m$ identity and $\mathbf I$ has context-implied
dimension. The operator $\operatorname{diag}(\cdot)$ forms a diagonal matrix, and $\circ$ denotes the
Hadamard product; $\mathbf z^{\circ2}$ denotes elementwise squaring. Here $\delta_x$ is a Dirac mass (distinct from aspect ratio $\delta$),
$D_{\mathbf z}$ is a Jacobian with respect to $\mathbf z$, and
$D_{\mathbf a}f:=\mathbf a^{\mathsf T}\nabla_{\mathbf y}f$ is a scalar
directional derivative. Also, $\sigma(\cdot)$ generates a sigma-field,
$\mathcal F\vee\mathcal G$ is a sigma-field join, and $\perp\!\!\!\perp$
denotes independence. We use $O_{\mathrm P}(\cdot)$ and $o_{\mathrm P}(\cdot)$ in their standard
stochastic senses, $\xrightarrow{\mathrm P}$ for convergence in probability, and
$\Rightarrow$ for weak convergence. Probability statements for a uniformly
sampled coordinate include the randomness of that coordinate.

%% file: sbl_tsp_body.tex
% Technical body v4 for IEEE TSP revision (Sections II--V only)
% Insert after the Introduction in the final IEEEtran manuscript.

\section{Problem Formulation and SBL Stationary Points}
\label{sec:model}

This section specifies stationary points of the finite-dimensional SBL evidence objective,
the local regularity required under finite-coordinate perturbations, and the
adaptive response induced by hyperparameter reoptimization.

\subsection{Gaussian Linear Model}

Consider a sequence of real-valued linear inverse problems
\begin{align}
    \mathbf y&=\mathbf A\mathbf x+\mathbf v, &\mathbf v\sim\mathcal N(\mathbf 0,\sigma_0^2\mathbf I_m),
    \label{eq:model}\\
    A_{\ell i}&\overset{\mathrm{i.i.d.}}{\sim}\mathcal N(0,1/m),&
    \frac{m}{n}\longrightarrow\delta\in(0,1).
    \label{eq:aspect_ratio}
\end{align}
Here $\mathbf y\in\mathbb R^m$, $\mathbf x\in\mathbb R^n$, and
$\mathbf A\in\mathbb R^{m\times n}$, with $1\le\ell\le m$ and $1\le i\le n$; $\mathbf v$ is independent of
$\mathbf A$ and $\sigma_0^2>0$. The unknown signal $\mathbf x=\mathbf x_n$ is deterministic, so probability statements concern only the Gaussian system and any explicit branch-selection randomization.

\begin{assumption}[Gaussian proportional regime]
\label{ass:model}
The sensing matrix and measurement noise satisfy
\eqref{eq:model}--\eqref{eq:aspect_ratio}. The deterministic signal obeys
$\sup_n\max_i|x_i|<\infty$, and its empirical distribution
$n^{-1}\sum_i\delta_{x_i}$ converges weakly to a probability law $P_X$.
\end{assumption}

No sparsity is assumed: an atom of $P_X$ at zero is needed only for the null law in Section~V-B.
The Gaussian design is essential in Section~\ref{sec:scalar}, where deletion gives exact
conditional independence and Gaussianity.

\subsection{SBL Evidence Objective and KKT Conditions}

In the working SBL model, the coefficients are conditionally independent with
$x_i\mid\gamma_i\sim\mathcal N(0,\gamma_i)$, and the variance hyperparameters
$\boldsymbol\gamma\geq\mathbf 0$ are estimated by marginal-likelihood (evidence)
maximization via type-II maximum likelihood \cite{Tipping2001,WipfRao2004}.
No hyperprior is imposed. Equivalently,
$\mathbf x\mid\boldsymbol\gamma\sim\mathcal N(\mathbf 0,\boldsymbol\Gamma)$,
where $\boldsymbol\Gamma=\operatorname{diag}(\boldsymbol\gamma)$, and the
working likelihood is
$p(\mathbf y\mid\mathbf x;\tau)=\mathcal N(\mathbf A\mathbf x,\tau\mathbf I_m)$.
This working prior is not a generative assumption on the deterministic true
signal. The model noise variance $\tau>0$ is fixed through Section~IV and
estimated jointly in Section~V; the true measurement-noise variance is
$\sigma_0^2$.

Marginalizing the coefficient vector gives
$p(\mathbf y\mid\boldsymbol\gamma,\tau)=\mathcal N(\mathbf 0,\mathbf C)$ with
$\mathbf C=\tau\mathbf I_m+\mathbf A\boldsymbol\Gamma\mathbf A^{\mathsf T}$.
After dropping the additive constant and the common factor $1/2$, evidence
maximization is equivalent to minimizing
\begin{equation}
    \mathcal L_n(\boldsymbol\gamma;\tau)
    =\log\det\mathbf C+\mathbf y^{\mathsf T}\mathbf C^{-1}\mathbf y,
    \qquad \boldsymbol\gamma\geq\mathbf 0.
    \label{eq:evidence}
\end{equation}
Define the resolvent and residual as
\begin{equation}
    \mathbf Q=\mathbf C^{-1},\qquad \mathbf r=\mathbf Q\mathbf y.
    \label{eq:Qr}
\end{equation}
The posterior mean under the working SBL model and the local quantities used throughout are
\begin{equation}
\begin{aligned}
    \boldsymbol\mu&=\boldsymbol\Gamma\mathbf A^{\mathsf T}\mathbf Q\mathbf y,\qquad \mu_i=\gamma_i p_i,\\
    p_i&=\mathbf a_i^{\mathsf T}\mathbf r,\qquad
    d_i=\mathbf a_i^{\mathsf T}\mathbf Q\mathbf a_i.
\end{aligned}
    \label{eq:local_quantities}
\end{equation}
The $i$th gradient component is
$\partial\mathcal L_n/\partial\gamma_i=d_i-p_i^2$, so the KKT conditions are
\begin{align}
    \gamma_i>0&\ \Longrightarrow\ d_i=p_i^2,
    \label{eq:kkt_active}\\
    \gamma_i=0&\ \Longrightarrow\ d_i-p_i^2\geq0.
    \label{eq:kkt_inactive}
\end{align}
We characterize these KKT points rather than the trajectory of a
particular SBL optimization algorithm.

\subsection{Regular Stationary-Point Class}

Because the SBL evidence objective is nonconvex, different algorithms or
initializations may select different stationary points
\cite{FaulTipping2001,WipfRao2004}. The cavity construction
therefore has to specify which local KKT component is followed when a coordinate
is deleted or perturbed; neither uniqueness nor global optimality is assumed.

\begin{definition}[Compatible stationary branch family]
\label{def:branch}
An admissible branch family consists of the selected full KKT point
and the auxiliary stationary points obtained by deleting, restoring, or
continuously perturbing finitely many coordinates in the constructions below.
These auxiliary points are selected measurably from the information available
to the corresponding problem and, along each prescribed homotopy, are joined
to its designated endpoints by a continuous path of KKT points. If
random tie-breaking is needed, the same system-independent auxiliary
randomization is used throughout; a shared auxiliary endpoint (such as the
pseudo-response state) denotes the same selected stationary point across the
homotopies that meet there. On a parameterized path, KKT stationarity is imposed
only on the hyperparameters that remain free; any explicitly controlled
coordinate is treated as an external parameter.

Compatibility is required only for the uniformly sampled coordinate $I$ and,
for the empirical argument, the sampled pair $(I,J)$ with probability tending
to one. No uniqueness, global optimality, or a priori proximity between full
and reduced stationary points is assumed.
\end{definition}

Definition~\ref{def:branch} restricts existence/selection of these finite-coordinate
continuations but assumes neither stability nor full-to-reduced proximity; local regularity is stated next.

\begin{assumption}[Local regularity of the selected KKT branch]
\label{ass:regular}
There exist deterministic constants $\gamma_{\max}<\infty$ and $c_H>0$,
independent of $n$, such that $0\leq\gamma_i\leq\gamma_{\max}$ along the
selected full, reduced, and parameterized branches. Let $\mathcal S$ denote
the active coordinates among the KKT variables that are free in the problem
under consideration, let $\mathbf A_{\mathcal S}$ collect the corresponding
columns, and define
$\mathbf p_{\mathcal S}=(p_i)_{i\in\mathcal S}$,
$\mathbf G_{\mathcal S}=\mathbf A_{\mathcal S}^{\mathsf T}\mathbf Q\mathbf A_{\mathcal S}$,
and $\mathbf D_{\mathbf p}=\operatorname{diag}(\mathbf p_{\mathcal S})$. The active Hessian of the evidence objective is
\begin{equation}
    \mathbf H_{\mathcal S}
    =2\mathbf D_{\mathbf p}\mathbf G_{\mathcal S}\mathbf D_{\mathbf p}
     -\mathbf G_{\mathcal S}\circ\mathbf G_{\mathcal S}.
    \label{eq:Hactive}
\end{equation}
On every smooth branch segment and its regular one-sided limit at an adjacent
support transition,
$\lambda_{\min}(\mathbf H_{\mathcal S})\geq c_H$. Support changes are locally
finite, isolated, and simple, with one free coordinate entering or leaving at
a time. For fixed sensing matrix and branch-selection randomization, the
selected residual maps are continuous and locally Lipschitz and are $C^1$
away from their regular support-transition boundaries; these boundaries are
of Gaussian measure zero. For the full selected map, this regularity is understood for almost every Gaussian observation.
\end{assumption}

Assumption~\ref{ass:regular} imposes uniform positive curvature on each free active face,
ensuring the nonsingularity needed for implicit continuation; it assumes neither sparsity,
a restricted-isometry property, nor full-to-reduced proximity. The map regularity supports
Gaussian integration by parts in Appendix~C.

\subsection{Adaptive Response at an SBL Stationary Point}

The KKT conditions describe the stationary point, but the cavity transfer also
requires its response to an observation perturbation while the active
hyperparameters are reoptimized. To distinguish the frozen and adaptive
responses, define
\begin{equation}
\begin{aligned}
    \eta_n&=\frac1m\operatorname{tr}\mathbf Q,\qquad
    V_n=\frac1m\|\mathbf r\|^2,\\
    \mathbf J&=D_{\mathbf y}\mathbf r,\qquad
    h_n=\frac1m\operatorname{tr}\mathbf J.
\end{aligned}
    \label{eq:macro_defs}
\end{equation}
Here $\eta_n$ is the normalized frozen-resolvent trace, $V_n$ is the normalized
residual energy, and $h_n$ is the normalized adaptive response after the active
hyperparameters are reoptimized (with $\tau$ fixed).

\begin{proposition}[Local SBL response]
\label{prop:local_response}
Under Assumption~\ref{ass:regular}, on every smooth segment,
\begin{align}
    D_{\mathbf y}\boldsymbol{\gamma}_{\mathcal S}
    &=2\mathbf H_{\mathcal S}^{-1}
      \mathbf D_{\mathbf p}\mathbf A_{\mathcal S}^{\mathsf T}\mathbf Q,
    \label{eq:Dgamma}\\
    \mathbf J
    &=\mathbf Q
    -2\mathbf Q\mathbf A_{\mathcal S}\mathbf D_{\mathbf p}
      \mathbf H_{\mathcal S}^{-1}\mathbf D_{\mathbf p}
      \mathbf A_{\mathcal S}^{\mathsf T}\mathbf Q.
    \label{eq:J_formula}
\end{align}
Consequently,
\begin{equation}
    h_n
    =\eta_n-\frac{2}{m}\operatorname{tr}\!\left(
      \mathbf H_{\mathcal S}^{-1}\mathbf L_{\mathcal S}
    \right),
    \label{eq:h_trace}
\end{equation}
where
$\mathbf L_{\mathcal S}=\mathbf D_{\mathbf p}\mathbf A_{\mathcal S}^{\mathsf T}
\mathbf Q^2\mathbf A_{\mathcal S}\mathbf D_{\mathbf p}$, and hence
$h_n\leq\eta_n$.
\end{proposition}

\emph{Proof:} See Appendix~A.

Proposition~\ref{prop:local_response} identifies the adaptive-response correction:
$h_n$ differs from the frozen-resolvent trace $\eta_n$ by the KKT-response term in
\eqref{eq:h_trace}, which becomes the coefficient correction in the scalar law.
The regularity bounds give $\eta_n,V_n,h_n=O_{\mathrm P}(1)$, while the resolvent bound
and Gaussian concentration keep $\eta_n$ and $V_n$ bounded away from zero in
probability along regular subsequences.

\section{Reoptimization-Corrected Scalar Law}
\label{sec:scalar}

We now characterize a typical coordinate after the remaining variance hyperparameters are reoptimized. A leave-one-out resolvent built from full-data hyperparameters is
not independent of the deleted column, so we start from a genuine reduced leave-one-out SBL
problem and transfer its exact conditional Gaussian law back to the selected full branch.

\subsection{Genuine One-Column Cavity}

For coordinate $i$, let $\mathbf A_{-i}$ and $\mathbf x_{-i}$ denote deletion
of the $i$th column and entry, and remove the corresponding true signal
contribution from the observation:
\begin{equation}
    \mathbf y_{-i}:=\mathbf y-\mathbf a_ix_i
    =\mathbf A_{-i}\mathbf x_{-i}+\mathbf v.
    \label{eq:reduced_observation}
\end{equation}
Let $\boldsymbol\gamma_{i,-}$ be the selected stationary hyperparameter vector
for $(\mathbf A_{-i},\mathbf y_{-i})$ and set
$\boldsymbol\Gamma_{i,-}:=\operatorname{diag}(\boldsymbol\gamma_{i,-})$. If the branch
selection uses the auxiliary randomization of Definition~\ref{def:branch},
denote it here by $\zeta$; it is independent of the Gaussian system and is the
same for all compatible full and reduced selections. Define
\begin{equation}
\begin{aligned}
    \mathbf Q_{i,-}&=\bigl(\tau\mathbf I_m+\mathbf A_{-i}\boldsymbol\Gamma_{i,-}\mathbf A_{-i}^{\mathsf T}\bigr)^{-1},\\
    \mathbf r_{i,-}&=\mathbf Q_{i,-}\mathbf y_{-i},\qquad
    \mathcal F_{-i}=\sigma(\mathbf A_{-i},\mathbf y_{-i},\zeta).
\end{aligned}
    \label{eq:reduced_objects}
\end{equation}
The entire reduced branch is $\mathcal F_{-i}$-measurable, while the deleted
Gaussian column is independent of this sigma-field:
\begin{equation}
    \mathbf a_i\perp\!\!\!\perp\mathcal F_{-i}.
    \label{eq:fresh_column}
\end{equation}
Set $\eta_{i,-}=m^{-1}\operatorname{tr}\mathbf Q_{i,-}$,
$V_{i,-}=m^{-1}\|\mathbf r_{i,-}\|^2$, and let $h_{i,-}$ be the normalized
trace of the reduced-branch residual Jacobian at $\mathbf y_{-i}$. The cavity
projection $Z_{i,-}=\mathbf a_i^{\mathsf T}\mathbf r_{i,-}$ is therefore exactly Gaussian conditional on $\mathcal F_{-i}$:
\begin{equation}
    Z_{i,-}\mid\mathcal F_{-i}\sim\mathcal N(0,V_{i,-}).
    \label{eq:exact_cavity_gaussian}
\end{equation}

This exact conditional Gaussian law still belongs to the deleted problem. Appendix~B returns
to the selected full branch through a response continuation
\begin{equation}
    \mathbf u_I(t)=\mathbf y_{-I}+t\mathbf a_I,
    \qquad |t|\leq B,
    \label{eq:response_path}
\end{equation}
together with insertion and local correction continuations. Deletion supplies
exact independence; these continuations supply the transfer back to the full
stationary point.

\begin{assumption}[Sublinear support-transition count]
\label{ass:continuation}
Let $I$ be uniform on $\{1,\ldots,n\}$ and independent of the system. For each
fixed $B>0$, let $N_{I,B}$ be the total number of support transitions of
the free KKT coordinates encountered on the response path
\eqref{eq:response_path} and on the one-coordinate insertion continuation of
Appendix~B. We assume
\begin{equation}
    \frac{N_{I,B}}{m}\xrightarrow{\mathrm P}0.
    \label{eq:nonextensive_transitions}
\end{equation}
The same condition holds after deletion of one additional uniformly sampled
coordinate, as required by the two-coordinate cavity argument.
\end{assumption}

Assumption~\ref{ass:continuation} is pathwise regularity, not a transfer or sparsity
condition: ``support'' means $\{i:\gamma_i>0\}$. It requires $N_{I,B}=o_{\mathrm P}(m)$;
Appendix~B derives the vanishing local correction while the total cavity-to-full nuisance
displacement may remain $O_{\mathrm P}(1)$. Their finite-dimensional counterparts can be checked through continuation
existence, active-Hessian and hyperparameter bounds, and transition counts; we
do not claim universal validity across algorithms or initializations.

\subsection{Reoptimization-Corrected Scalar Gaussian Law}

\begin{theorem}[Reoptimization-corrected scalar Gaussian law]
\label{thm:scalarization}
For a branch family satisfying Definition~\ref{def:branch} and
Assumptions~\ref{ass:model}--\ref{ass:continuation}, let $I$ be uniform on
$\{1,\ldots,n\}$ and independent of the system. Then
\begin{align}
    &p_I+h_{I,-}\mu_I
    =h_{I,-}x_I+\sqrt{V_{I,-}}\,G_I+o_{\mathrm P}(1),
    \label{eq:quenched_scalar}\\
    &\eta_{I,-}-\eta_n=o_{\mathrm P}(1), \quad
    V_{I,-}-V_n=o_{\mathrm P}(1),\\
    &h_{I,-} -h_n=o_{\mathrm P}(1).
    \label{eq:cavity_macro_replace}\\
    &U_i:=p_i+h_n\mu_i,
    \label{eq:reaction_field}\\
    &U_I=h_nx_I+\sqrt{V_n}\,G_I+o_{\mathrm P}(1).
    \label{eq:full_scalar_law}
\end{align}
Here $G_I\sim\mathcal N(0,1)$ conditionally on the genuine cavity environment.
Furthermore, for every bounded Lipschitz
$\varphi:\mathbb R^2\to\mathbb R$,
\begin{equation}
    \frac1n\sum_{i=1}^{n}
    \left\{\varphi(x_i,U_i)-
    \mathbb E_G\!\left[\varphi\!\left(x_i,h_nx_i+\sqrt{V_n}G\right)\right]\right\}
    \xrightarrow{\mathrm P}0,
    \label{eq:empirical_scalarization}
\end{equation}
where $G\sim\mathcal N(0,1)$ is an auxiliary Gaussian variable independent
of the finite-dimensional system; $\mathbb E_G$ denotes expectation only over
this variable.
\end{theorem}

\emph{Proof:} See Appendix~B.

Theorem~\ref{thm:scalarization} shows that the high-dimensional adaptive stationary problem nevertheless decouples, for a typical coordinate, into a scalar Gaussian channel. The reoptimization effect does not vanish: it changes the effective signal gain from the frozen response $\eta_n$ to the adaptive response $h_n$. Equation~\eqref{eq:empirical_scalarization} upgrades this typical-coordinate statement to an empirical law, so the same scalar description governs the asymptotic coordinate population.

\subsection{Scalar KKT Geometry}

Combining Theorem~\ref{thm:scalarization} with the finite-dimensional KKT conditions, introduce the marked empirical law
\begin{equation}
    \nu_n=\frac1n\sum_{i=1}^{n}\delta_{(x_i,U_i,\gamma_i,p_i,\mu_i)}.
    \label{eq:marked_measure}
\end{equation}
Let $\tau_n=\tau$ in the fixed-noise analysis. Along any subsequence on which
$(\eta_n,h_n,V_n,\tau_n,\nu_n)$ converges jointly in distribution (with
$\nu_n$ endowed with the weak topology), denote the limiting random element by
$(\eta,h,V,\tau,\nu_b)$, where $b$ indexes the selected branch, and write
$\langle f\rangle_b=\int f\,d\nu_b$; identities hold almost surely under this
limit. Let $\mathcal M=(\eta,h,V,\tau)$. A generic mark is
$(X,U,\Gamma,P,\widehat X)$ with $\widehat X=\Gamma P$ and $P$ the limit of $p_i$.
Equation~\eqref{eq:empirical_scalarization}
then identifies its $(X,U)$ marginal as
\begin{equation}
    U=hX+\sqrt V\,G,\qquad
    G\sim\mathcal N(0,1),\qquad G\perp\!\!\!\perp X\mid\mathcal M,
    \label{eq:scalar_channel}
\end{equation}
with $X\sim P_X$. No independence is asserted between
$(\Gamma,P,\widehat X)$ and $(X,U)$.

To combine the scalar law with the KKT conditions, Appendix~C establishes
the typical-coordinate leave-one-out diagonal resolvent relation
\begin{equation}
    d_I-\frac{\eta_n}{1+\eta_n\gamma_I}=o_{\mathrm P}(1).
    \label{eq:di_selfscreen}
\end{equation}
The same relation holds at the empirical-measure level. Combining it with the
finite-dimensional KKT inequalities therefore yields the limiting scalar KKT
relation
\begin{align}
    \Gamma=0
    &\ \Longrightarrow\ \widehat X=0,\quad U=P,\quad U^2=P^2\leq\eta,
    \label{eq:scalar_inactive}\\
    \Gamma>0
    &\ \Longrightarrow\
    P^2=\frac{\eta}{1+\eta\Gamma}.
    \label{eq:scalar_active}
\end{align}
Moreover, every active scalar KKT point obeys
\begin{align}
    U&=P+h\widehat X=P(1+h\Gamma),\\
    \frac{U^2}{\eta}
    &=\frac{(1+h\Gamma)^2}{1+\eta\Gamma}.
    \label{eq:scalar_kkt_graph}
\end{align}

Equation~\eqref{eq:scalar_kkt_graph} is generally a set-valued scalar KKT
relation, not an update rule: $\eta$ describes the frozen coordinate curvature,
whereas $h$ retains the adaptive-response correction. If several positive roots coexist,
the selected finite-dimensional branch determines the root represented in
$\nu_b$. In the frozen-response special case $h=\eta$, \eqref{eq:scalar_inactive}--
\eqref{eq:scalar_kkt_graph} reduce to
$\Gamma=(U^2-\eta)_+/\eta^2$, where $(a)_+:=\max\{a,0\}$, recovering
the fixed-point relation of the preliminary analysis in
\cite{XiaoSlock2026ICASSP}.

\section{Branchwise Large-System Characterization}
\label{sec:macro}

With $\tau$ fixed, the first two balances constrain the frozen-resolvent trace and
noise coupling, while the third closes the adaptive response.
The marked law retains branch-dependent root selection.

\subsection{Resolvent and Noise Balances}

\begin{theorem}[Branchwise large-system balances]
\label{thm:macro_balances}
Under the setting of Theorem~\ref{thm:scalarization}, every joint
subsequential limit $(\eta,h,V,\tau,\nu_b)$ of
$(\eta_n,h_n,V_n,\tau_n,\nu_n)$ satisfies, almost surely,
\begin{align}
    1
    &=\tau\eta+\frac1\delta
      \left\langle\frac{\eta\Gamma}{1+\eta\Gamma}\right\rangle_b,
    \label{eq:balance_eta}\\
    \sigma_0^2h
    &=\tau V+\frac1\delta
      \left\langle(\widehat X-X)P\right\rangle_b.
    \label{eq:balance_noise}
\end{align}
\end{theorem}

\emph{Proof:} See Appendix~C.

The first balance follows from the exact trace identity
\begin{equation}
    1=\tau\eta_n+\frac1m\sum_{i=1}^{n}\mu_ip_i.
    \label{eq:exact_trace_balance}
\end{equation}
For the second, the posterior-mean identity and the observation model give
\begin{align}
    \mathbf y&=\mathbf A\boldsymbol\mu+\tau\mathbf r,
    \label{eq:reconstruction_identity}\\
    \mathbf v&=\mathbf A(\boldsymbol\mu-\mathbf x)+\tau\mathbf r.
    \label{eq:noise_identity}
\end{align}
Because $\mathbf r$ is adaptive, the mixed noise term cannot be conditionally centered. Appendix~C instead proves by Gaussian integration by parts
\cite{Stein1981,BellecZhang2021} that
\begin{equation}
    m^{-1}\mathbf v^{\mathsf T}\mathbf r-\sigma_0^2h_n\xrightarrow{\mathrm P}0.
    \label{eq:stein_noise}
\end{equation}
This yields the branchwise noise balance; closing $h$ requires the directional KKT response below.

\subsection{Adaptive-Response Balance}

Closing the adaptive response additionally differentiates the reduced
response at its random endpoint, so we exclude the asymptotically exceptional
case in which that endpoint hits a support-transition kink.

\begin{assumption}[Typical endpoint differentiability]
\label{ass:derivative}
Along the reduced response path \eqref{eq:response_path}, let
$\mathcal T_{I,B}$ be the set of support-transition points on $[-B,B]$ and let
$e_I:=x_I-\mu_I$ be the full-branch response endpoint. Write
$\operatorname{dist}(t,\mathcal T):=\inf_{s\in\mathcal T}|t-s|$ and
$\operatorname{dist}(t,\varnothing):=\infty$. There exists a deterministic
sequence $\omega_n\downarrow0$
with $\log(1/\omega_n)=o(m)$ such that, for every sufficiently large fixed
$B$,
\begin{equation}
    \Pr\!\left(
      |e_I|\leq B,\;
      \operatorname{dist}(e_I,\mathcal T_{I,B})\leq\omega_n
    \right)\longrightarrow0.
    \label{eq:endpoint_away_kink}
\end{equation}
In addition, on $\{|e_I|\leq B,\gamma_I>0\}$, the probability that the
regular correction continuation of Appendix~B crosses a nuisance
support-transition boundary tends to zero.
\end{assumption}

Assumption~\ref{ass:derivative} is used only for derivative-level adaptive-response
transfer and the closed reconstruction-error relation, not for Theorem~\ref{thm:scalarization}
or the first two balances.

Define $\Psi_{\eta,h}(\Gamma)=0$ at $\Gamma=0$ and
$\Psi_{\eta,h}(\Gamma)=(2+\eta\Gamma)/[(2h-\eta)+\eta h\Gamma]$ for
$\Gamma>0$; Appendix~D proves that the active denominator is nonzero.

\begin{theorem}[Adaptive-response balance]
\label{thm:susceptibility}
Under the setting of Theorem~\ref{thm:scalarization} and
Assumption~\ref{ass:derivative}, every joint subsequential limit
$(\eta,h,V,\tau,\nu_b)$ additionally satisfies, almost surely,
\begin{equation}
\begin{aligned}
    1
    ={}&\tau h
    +\frac{h}{\delta}
     \left\langle\Psi_{\eta,h}(\Gamma)\right\rangle_b.
\end{aligned}
    \label{eq:balance_h}
\end{equation}
\end{theorem}

\emph{Proof:} See Appendix~D.

Appendix~D evaluates the directional posterior-mean response after hyperparameter
reoptimization and takes its normalized trace. Assumption~\ref{ass:derivative} is used only
at the random endpoint; a finite-dimensional Schur-complement gap keeps the active denominator
nonzero. The balance implies $h\neq0$ for every regular limiting state.

Together, Theorems~\ref{thm:scalarization}--\ref{thm:susceptibility} reduce the high-dimensional stationary KKT system to a scalar Gaussian channel, a scalar KKT graph, and three macroscopic branchwise balances. The state variables have distinct roles: $\eta$ describes frozen-resolvent geometry, $V$ the scalar fluctuation scale, and $h$ the response after hyperparameter reoptimization. Theorem~\ref{thm:susceptibility} closes the adaptive quantity $h$, which cannot in general be replaced by $\eta$; no uniqueness of stationary branches or scalar roots is required.

\section{Joint Noise Variance Estimation and Reconstruction Error}
\label{sec:noise}

We now estimate the model noise variance jointly with the variance
hyperparameters. Interior stationarity gives the exact identity $V_n=\eta_n$; under the joint
regularity below, the extra response from estimating this parameter is lower order, so the leading scalar and large-system laws persist.

\subsection{Noise-Variance Stationarity}

For fixed $\tau$, an overbar denotes variance-hyperparameter profiling. Let
$\overline{\boldsymbol{\gamma}}(\mathbf y,\tau)$ denote the selected branch and
$\overline{\mathcal L}(\mathbf y,\tau)
:=\mathcal L_n(\overline{\boldsymbol{\gamma}}(\mathbf y,\tau);\tau)$.
The derivatives of the unprofiled and profiled objectives with respect to $\tau$ are
\begin{equation}
\begin{aligned}
    s_\tau&:=\partial_\tau\mathcal L_n
    =\operatorname{tr}\mathbf Q-\mathbf y^{\mathsf T}\mathbf Q^2\mathbf y,\\
    S_\tau&:=s_\tau(\overline{\boldsymbol\gamma};\mathbf y,\tau)
    =\partial_\tau\overline{\mathcal L},
\end{aligned}
    \label{eq:noise_scores}
\end{equation}
where the last equality is the envelope theorem on each smooth regular
segment.

\begin{proposition}[Noise-variance stationarity]
\label{prop:self_normalization}
At every interior joint stationary point with respect to
$(\boldsymbol{\gamma},\tau)$,
\begin{equation}
    V_n=\eta_n.
    \label{eq:finite_self_normalization}
\end{equation}
\end{proposition}

\begin{IEEEproof}
At an interior joint stationary point, $S_\tau=0$. Since
$\mathbf r=\mathbf Q\mathbf y$,
\begin{equation}
    0=\operatorname{tr}\mathbf Q-\|\mathbf r\|^2
     =m(\eta_n-V_n).
\end{equation}
\end{IEEEproof}

Proposition~\ref{prop:self_normalization} is exact and finite-dimensional. With fixed $\tau$, $V_n$ and $\eta_n$ are distinct residual-energy and resolvent-response states; interior evidence stationarity in $\tau$ forces them to coincide exactly, removing one macroscopic degree of freedom. Extending the cavity and adaptive-response results to data-dependent $\tau$ requires positive profiled curvature.

\begin{assumption}[Regular jointly profiled noise branch]
\label{ass:joint_regular}
There exist deterministic constants
$0<\tau_{\min}\leq\tau_{\max}<\infty$ and $c_\tau>0$, independent of $n$,
such that the jointly selected full, reduced, response, insertion, and
correction branches remain on the compatible regular family of
Definition~\ref{def:branch}, with $\tau_{\min}\leq\tau\leq\tau_{\max}$.
Along these paths, the free hyperparameters satisfy
Assumption~\ref{ass:regular}; the support-transition count condition of
Assumption~\ref{ass:continuation} and, when derivatives at the random endpoint
are used, Assumption~\ref{ass:derivative}, hold for the corresponding jointly
profiled and double-reduced paths. On every smooth segment, with the regular
one-sided limit at a simple support transition,
\begin{equation}
    \frac1m\frac{\partial S_\tau}{\partial\tau}\geq c_\tau,
    \label{eq:tau_curvature}
\end{equation}
where $S_\tau$ profiles only the hyperparameters that are free on the current
path; any explicitly controlled coordinate is held fixed.
\end{assumption}

Assumption~\ref{ass:joint_regular} extends the preceding branch regularity to jointly profiled
paths and adds boundedness and uniform positive curvature of the profiled $\tau$ equation.

Appendix~E analyzes the jointly profiled response directly rather than
substituting a random estimated $\tau$ into the fixed-$\tau$ theory. Estimating $\tau$ adds only a rank-one response whose normalized trace is
$O_{\mathrm P}(m^{-1})$, and, for a uniformly sampled deletion, the full and reduced profiled $\tau$ values are
$o_{\mathrm P}(m^{-1/2})$ apart. Hence the fixed-$\tau$ profiled and fully joint
adaptive responses have the same limit, so
Theorems~\ref{thm:scalarization}--\ref{thm:susceptibility} extend to the joint
branch. The only additional limiting constraint is
\begin{equation}
    V=\eta.
    \label{eq:V_equals_eta}
\end{equation}

\subsection{Normalized Scalar Distribution}

For a regular joint-estimation limit, $\eta>0$ and $V=\eta$. Under its limiting
marked law, define $Y=U/\sqrt\eta$ and $\alpha=h/\sqrt\eta$. For $P_X$-almost every $x$, the squared normalized field satisfies the first relation below;
when $P_X(\{0\})>0$, its null specialization is the second:
\begin{align}
    Y^2\mid(X=x,\mathcal M)&\sim\chi_1^2(\alpha^2x^2),
    \label{eq:noncentral_chi}\\
    Y^2\mid X=0&\sim\chi_1^2.
    \label{eq:null_chi}
\end{align}
Here $\chi_1^2(\lambda)$ denotes the noncentral chi-square distribution with
one degree of freedom and noncentrality parameter $\lambda$
($\chi_1^2=\chi_1^2(0)$). Thus the null law, when present, is parameter free: joint noise stationarity self-normalizes the scalar field, since $X=0$ gives $U/\sqrt\eta=G$. The preceding characterization does not require an atom at zero; for nonzero coordinates the noncentrality remains branch dependent through $\alpha$.

\subsection{Reconstruction Error}

Finally, consider the mean-square error (MSE) of the posterior mean,
\begin{equation}
    \operatorname{MSE}_n=\frac1n\|\boldsymbol\mu-\mathbf x\|^2.
    \label{eq:mse_definition}
\end{equation}
Joint noise-variance estimation is not needed for the basic MSE identity. Along a
subsequence covered by Theorems~\ref{thm:scalarization}--\ref{thm:susceptibility},
let $\mathcal E_b=\langle(\widehat X-X)^2\rangle_b$. Uniform integrability
identifies $\mathcal E_b$ as the subsequential limit of the empirical MSE;
when the limiting marked law is deterministic, the convergence is in
probability. Theorem~\ref{thm:susceptibility} also gives $h\neq0$.

\begin{corollary}[Effective-noise and MSE relation]
\label{cor:mse}
Under the setting of Theorem~\ref{thm:susceptibility}, suppose the selected scalar branch admits
$\widehat X=\widehat x_b(X,U;\mathcal M)$ such that, for $P_X$-almost every $x$, the map in $u$ is continuous, piecewise $C^1$ with locally finitely many breakpoints, and has integrable derivative under the scalar Gaussian law. This additional condition selects a differentiable branch when several positive KKT roots coexist. Then every such subsequential limit satisfies, almost surely,
\begin{equation}
    \mathcal E_b
    =\delta\left(\frac{V}{h^2}-\sigma_0^2\right).
    \label{eq:closed_mse}
\end{equation}
If, in addition, the model noise variance is jointly estimated under Assumption~\ref{ass:joint_regular}, then $V=\eta$ and
\begin{align}
    \mathcal E_b
    &=\delta\left(\frac{\eta}{h^2}-\sigma_0^2\right),
    \label{eq:closed_mse_learned}\\
    \alpha^{-2}
    &=\sigma_0^2+\frac{\mathcal E_b}{\delta}.
    \label{eq:alpha_mse}
\end{align}
\end{corollary}

\emph{Proof:} See Appendix~E.

Appendix~E combines the scalar-branch derivative with scalar Stein's identity
\cite{Stein1981} and the noise and adaptive-response balances to obtain
Corollary~\ref{cor:mse}. Rescaling the scalar channel gives
$U/h=X+(\sqrt V/h)G$, so $V/h^2$ is its effective noise variance. Equation~\eqref{eq:closed_mse} shows that this variance decomposes into the physical noise $\sigma_0^2$ plus the branchwise reconstruction error $\mathcal E_b/\delta$; under joint noise-variance estimation, $V=\eta$ expresses the same relation through $(\eta,h)$.

\section{Conclusion}
\label{sec:conclusion}

We developed a large-system characterization of regular stationary branches of the classical SBL evidence objective under i.i.d.\ Gaussian sensing. Hyperparameter reoptimization leaves a nonvanishing macroscopic feedback, so the effective scalar channel is governed by the adaptive response $h$ rather than the frozen resolvent response $\eta$. The high-dimensional stationary system is thereby reduced to a scalar Gaussian law, a generally set-valued scalar KKT graph, and three branchwise balance equations. Joint evidence learning of the noise variance further enforces $V=\eta$, yielding a self-normalized chi-square null law when $P_X(\{0\})>0$, while $V/h^2=\sigma_0^2+\mathcal E_b/\delta$ links the equivalent scalar noise directly to reconstruction error. Establishing when standard SBL algorithms select compatible regular branches, and extending the analysis beyond i.i.d.\ Gaussian sensing, remain important directions for future work.

%% file: sbl_tsp_appendices.tex
% Appendices v26: compressed reviewer-facing proofs
\appendices

\section{Local Geometry of Regular SBL Stationary Points}
\label{app:local}

This appendix proves Proposition~\ref{prop:local_response} and collects the
moment bounds used in the cavity and Stein arguments. From \eqref{eq:Qr},
\begin{equation}
\begin{aligned}
\partial_{\gamma_j}\mathbf Q&=-\mathbf Q\mathbf a_j\mathbf a_j^{\mathsf T}\mathbf Q,\\
\partial_{\gamma_j}d_i&=-G_{ij}^2,\qquad
\partial_{\gamma_j}p_i=-G_{ij}p_j,
\end{aligned}
\end{equation}
where $G_{ij}=\mathbf a_i^{\mathsf T}\mathbf Q\mathbf a_j$. Hence the
Jacobian of the active KKT residual $\mathbf g_{\mathcal S}=(d_i-p_i^2)_{i\in\mathcal S}$
with respect to $\boldsymbol\gamma_{\mathcal S}$ is
$\mathbf H_{\mathcal S}$ in \eqref{eq:Hactive}, while
$D_{\mathbf y}\mathbf g_{\mathcal S}=-2\mathbf D_{\mathbf p}
\mathbf A_{\mathcal S}^{\mathsf T}\mathbf Q$. Implicit differentiation gives
\eqref{eq:Dgamma}; substituting it into
$d\mathbf r=\mathbf Q\,d\mathbf y-\mathbf Q\mathbf A_{\mathcal S}
\mathbf D_{\mathbf p}\,d\boldsymbol\gamma_{\mathcal S}$ gives
\eqref{eq:J_formula} and, after taking traces, \eqref{eq:h_trace}. Since
$\mathbf H_{\mathcal S}^{-1}\succ0$ and $\mathbf L_{\mathcal S}\succeq0$,
$h_n\le\eta_n$.

Let $\tau_\star=\tau$ in the fixed-noise analysis and
$\tau_\star=\tau_{\min}$ in the joint-noise analysis. The KKT inequalities and
resolvent bound imply
\begin{equation}
\|\mathbf Q\|_{\mathrm{op}}\le\tau_\star^{-1},
\qquad
p_i^2\le d_i\le\tau_\star^{-1}\|\mathbf a_i\|^2.
\label{eq:A_basic_bounds}
\end{equation}
Indeed, \eqref{eq:J_formula}, \eqref{eq:A_basic_bounds}, and the Hessian gap give
$\|\mathbf J\|_{\mathrm{op}}=O_{\mathrm P}(1)$ because
$\|\mathbf D_{\mathbf p}\|_{\mathrm{op}}\le\tau_\star^{-1/2}\max_i\|\mathbf a_i\|$ and all remaining factors have bounded operator norm. Since
$\|\mathbf J\|_{\mathrm F}^2\le m\|\mathbf J\|_{\mathrm{op}}^2$, standard Gaussian spectral-, column-, and residual-norm bounds \cite{Vershynin2018}, uniformly over each fixed bounded regular continuation, give
\begin{equation}
\begin{aligned}
\mathbb E\|\mathbf r\|^2=O(m),\ 
\mathbb E\|\mathbf J\|_{\mathrm F}^2=O(m),\ 
\sup_n\mathbb E\|\mathbf J\|_{\mathrm{op}}^2<\infty.
\end{aligned}
\label{eq:A_moment_bounds}
\end{equation}

\begin{lemma}[Adaptive residual regularity]
\label{lem:sobolev_residual}
Under Assumption~\ref{ass:regular}, for almost every fixed $(\mathbf A,\zeta)$
the selected residual map belongs to the Gaussian Sobolev class $W^{1,2}$
with respect to the observation Gaussian measure, and its weak Jacobian agrees
almost everywhere with \eqref{eq:J_formula}.
\end{lemma}
\begin{IEEEproof}
The selected map is locally Lipschitz and piecewise $C^1$, with transition boundaries of Gaussian measure zero. Rademacher's theorem identifies its weak derivative with the piecewise-$C^1$ derivative. Equation~\eqref{eq:A_moment_bounds} and Fubini
then give square integrability of both the residual and its weak Jacobian for
almost every fixed $(\mathbf A,\zeta)$.
\end{IEEEproof}

\section{Proof of the Reoptimization-Corrected Scalar Law}
\label{app:scalar}

The proof has four steps: deleted-column Gaussian concentration, insertion/correction transfer, integrated adaptive response, and a two-column cavity for empirical
self-averaging. All pathwise stochastic bounds below are uniform over the
smooth segments of a fixed bounded continuation; Assumptions~\ref{ass:regular}
and~\ref{ass:continuation} control the intervening simple support changes.
Appendix decorations are $\mathrm F$ (selected full point), $-i/-IJ$ (genuine
deletion), $\mathrm{ps}$ (pseudo-response), $\mathrm c$ (correction), $\mathrm R$ (algebraic
full-nuisance leave-one-out), and $\mathrm{jt}$ (joint profiling); matching $\boldsymbol\Gamma$ decorations denote the diagonal matrices formed from the corresponding $\boldsymbol\gamma$, and dots denote path derivatives.

\subsection{Reduced Cavity and Local Transfer}

By \eqref{eq:fresh_column},
$Z_{i,-}\mid\mathcal F_{-i}\sim\mathcal N(0,V_{i,-})$ exactly. Put
$\mathbf c_i^{(0)}=\mathbf A_{-i}^{\mathsf T}\mathbf Q_{i,-}\mathbf a_i$.
Set $\rho_n:=\sqrt{\log n/m}$. Conditionally on $\mathcal F_{-i}$,
\begin{equation}
\operatorname{Var}(c_{ji}^{(0)}\mid\mathcal F_{-i})
=m^{-1}\mathbf a_j^{\mathsf T}\mathbf Q_{i,-}^2\mathbf a_j=O_{\mathrm P}(m^{-1}).
\end{equation}
A Gaussian union bound over $j$, together with standard concentration bounds \cite{RudelsonVershynin2013,Vershynin2018}, and the operator bound on $\mathbf Q_{i,-}$ therefore give
\begin{equation}
\begin{aligned}
\|\mathbf c_i^{(0)}\|_\infty&=O_{\mathrm P}(\rho_n),&
\|\mathbf c_i^{(0)}\|_2&=O_{\mathrm P}(1),\\
\|(\mathbf c_i^{(0)})^{\circ2}\|_2&=O_{\mathrm P}(\rho_n).&&
\end{aligned}
\end{equation}
This rate comes from the Gaussian maximum bound over the columns.

For the insertion path
\begin{equation}
\begin{aligned}
\mathbf y_\lambda&=\mathbf y_{-i}+\lambda x_i\mathbf a_i,\\
\mathbf C_\lambda&=\tau\mathbf I_m+
\mathbf A_{-i}\boldsymbol\Gamma_{-i}(\lambda)\mathbf A_{-i}^{\mathsf T}
+\lambda\gamma_i^{\mathrm F}\mathbf a_i\mathbf a_i^{\mathsf T}.
\end{aligned}
\end{equation}
With $\mathbf Q_\lambda:=\mathbf C_\lambda^{-1}$,
$\mathbf r_\lambda:=\mathbf Q_\lambda\mathbf y_\lambda$, and $\mathcal S_\lambda$
the active nuisance set, put $p_i:=\mathbf a_i^{\mathsf T}\mathbf r_\lambda$,
$s_i:=\mathbf a_i^{\mathsf T}\mathbf Q_\lambda\mathbf a_i$, $e_i:=x_i-\gamma_i^{\mathrm F}p_i$, and
$\mathbf c_i=\mathbf A_{\mathcal S_\lambda}^{\mathsf T}\mathbf Q_\lambda\mathbf a_i$.
Differentiating the nuisance KKT equations gives
\begin{equation}
\mathbf H_{\mathcal S_\lambda}\dot{\boldsymbol\gamma}_{\mathcal S_\lambda}
=\gamma_i^{\mathrm F}\mathbf c_i^{\circ2}
+2e_i\mathbf D_{\mathbf p}\mathbf c_i,
\quad
\dot p_i=s_i e_i-(\mathbf D_{\mathbf p}\mathbf c_i)^{\mathsf T}
\dot{\boldsymbol\gamma}_{\mathcal S_\lambda}.
\end{equation}
The Hessian gap and Gronwall yield
$\sup_\lambda|p_i(\lambda)|+\sup_\lambda
\|\dot{\boldsymbol\gamma}_{\mathcal S_\lambda}\|_2=O_{\mathrm P}(1)$.
For the nuisance cross-coupling vector
$\mathbf c_{-i}=\mathbf A_{-i}^{\mathsf T}\mathbf Q_\lambda\mathbf a_i$, let
$\mathbf G_{-i,\lambda}:=\mathbf A_{-i}^{\mathsf T}\mathbf Q_\lambda\mathbf A_{-i}$.
Then
\begin{equation}
\dot{\mathbf c}_{-i}
=-\gamma_i^{\mathrm F}s_i\mathbf c_{-i}
-\mathbf G_{-i,\lambda}(\dot{\boldsymbol\gamma}_{-i}\circ
\mathbf c_{-i}).
\end{equation}
Hence, piecewise Gronwall propagates the initial Gaussian bound:
$\sup_{\lambda\in[0,1]}\|\mathbf c_{-i}(\lambda)\|_\infty=O_{\mathrm P}(\rho_n)$.

At the selected full point define
\begin{equation}
\mathbf R_i^{\mathrm F}=(\tau\mathbf I_m+\mathbf A_{-i}\boldsymbol\Gamma_{-i}^{\mathrm F}
\mathbf A_{-i}^{\mathsf T})^{-1},
\qquad
\mathbf u_i^{\mathrm{ps}}=\mathbf y-\mu_i^{\mathrm F}\mathbf a_i.
\end{equation}
Sherman--Morrison gives $\mathbf r^{\mathrm F}=\mathbf r=\mathbf R_i^{\mathrm F}\mathbf u_i^{\mathrm{ps}}$.
Rather than infer parameter proximity from a KKT residual of order $O_{\mathrm P}(\rho_n)$, we use the correction continuation
\begin{equation}
\mathbf y_\theta^{\mathrm c}=\mathbf u_i^{\mathrm{ps}}+\theta\mu_i^{\mathrm F}\mathbf a_i,
\quad
\mathbf C_\theta^{\mathrm c}=\tau\mathbf I_m+
\mathbf A_{-i}\boldsymbol\Gamma_{-i}^{\mathrm c}(\theta)\mathbf A_{-i}^{\mathsf T}
+\theta\gamma_i^{\mathrm F}\mathbf a_i\mathbf a_i^{\mathsf T}.
\end{equation}
With $\mathbf Q_\theta^{\mathrm c}:=(\mathbf C_\theta^{\mathrm c})^{-1}$,
$\mathbf r_\theta^{\mathrm c}:=\mathbf Q_\theta^{\mathrm c}\mathbf y_\theta^{\mathrm c}$,
$p_i^{\mathrm c}:=\mathbf a_i^{\mathsf T}\mathbf r_\theta^{\mathrm c}$, and $\mathcal S_\theta$ the
active nuisance set, put $\Delta_i:=p_i^{\mathrm F}-p_i^{\mathrm c}$,
$s_i^{\mathrm c}:=\mathbf a_i^{\mathsf T}\mathbf Q_\theta^{\mathrm c}\mathbf a_i$ and
$\mathbf c_i^{\mathrm c}=\mathbf A_{\mathcal S_\theta}^{\mathsf T}\mathbf Q_\theta^{\mathrm c}\mathbf a_i$,
its differentiated KKT system is
\begin{align}
\mathbf H_{\mathcal S_\theta}\dot{\boldsymbol\gamma}_{\mathcal S_\theta}^{\mathrm c}
&=\gamma_i^{\mathrm F}(\mathbf c_i^{\mathrm c})^{\circ2}
+2\gamma_i^{\mathrm F}\Delta_i\mathbf D_{\mathbf p}\mathbf c_i^{\mathrm c},
\label{eq:B_correction_KKT}\\
\dot p_i^{\mathrm c}
&=\gamma_i^{\mathrm F}s_i^{\mathrm c}\Delta_i-(\mathbf D_{\mathbf p}\mathbf c_i^{\mathrm c})^{\mathsf T}
\dot{\boldsymbol\gamma}_{\mathcal S_\theta}^{\mathrm c}.
\label{eq:B_correction_pi}
\end{align}
The insertion bound supplies $\|\mathbf c_i^{\mathrm c}(1)\|_\infty=O_{\mathrm P}(\rho_n)$.
The same cross-coupling differential equation as above, followed by a
bootstrap in \eqref{eq:B_correction_KKT}--\eqref{eq:B_correction_pi}, gives
\begin{equation}
\sup_\theta|\Delta_i(\theta)|=O_{\mathrm P}(\rho_n),
\qquad
\sup_\theta\|\dot{\boldsymbol\gamma}^{\mathrm c}(\theta)\|_2=O_{\mathrm P}(\rho_n).
\end{equation}
Integration over $\theta\in[0,1]$ therefore yields
\begin{align}
\|\boldsymbol\gamma_{-i}^{\mathrm{ps}}-\boldsymbol\gamma_{-i}^{\mathrm F}\|_2
&=O_{\mathrm P}(\rho_n),
\label{eq:B_gamma_transfer}\\
|\mathbf a_i^{\mathsf T}\mathbf r_i^{\mathrm{ps}}-p_i^{\mathrm F}|
&=O_{\mathrm P}(\rho_n)=o_{\mathrm P}(1).
\label{eq:B_projection_transfer}
\end{align}
Thus only the local pseudo-response nuisance vector is close to the full
nuisance vector; no closeness of the genuine cavity endpoint is assumed.

\subsection{Integrated Adaptive Response and Large-System State Replacement}

Let $\mathcal R_{-i}$ be the genuine reduced residual map. For
$\mathbf u_i(t):=\mathbf y_{-i}+t\mathbf a_i$, let $\mathbf Q_{i,-}(t)$ be its
resolvent and set $\mathbf r_i(t):=\mathcal R_{-i}(\mathbf u_i(t))$,
$\mathbf J_i(t):=D_{\mathbf u}\mathcal R_{-i}(\mathbf u_i(t))$, and
$h_i(t)=m^{-1}\operatorname{tr}\mathbf J_i(t)$. On every smooth segment,
\begin{equation}
\frac{d}{dt}[\mathbf a_i^{\mathsf T}\mathbf r_i(t)]
=\mathbf a_i^{\mathsf T}\mathbf J_i(t)\mathbf a_i.
\label{eq:B_response_derivative}
\end{equation}
The response-path KKT derivative is $O_{\mathrm P}(1)$ in $\ell_2$, so the same
Gronwall argument used above gives, for fixed $B$,
\begin{equation}
\sup_{|t|\le B}\|\mathbf A_{-i}^{\mathsf T}\mathbf Q_{i,-}(t)\mathbf a_i\|_\infty
=O_{\mathrm P}(\rho_n).
\end{equation}
For deterministic $t$, with
$\mathcal G_i(t)=\sigma(\mathbf A_{-i},\mathbf y_{-i}+t\mathbf a_i,\zeta)$,
Gaussian regression gives
\begin{align}
\mathbf a_i\mid\mathcal G_i(t)
&\sim\mathcal N(\mathbf m_i(t),\varsigma_i^2(t)\mathbf I_m),\\
\mathbf m_i(t)
&=\frac{t[\mathbf u_i(t)-\mathbf A_{-i}\mathbf x_{-i}]}
{m\sigma_0^2+t^2}.
\end{align}
Here $\varsigma_i(t)>0$ is defined by
$\varsigma_i^2(t)=\sigma_0^2/(m\sigma_0^2+t^2)$. Thus $\|\mathbf m_i(t)\|=O_{\mathrm P}(m^{-1/2})$ and
$\varsigma_i^2(t)=m^{-1}+O(m^{-2})$. Using \eqref{eq:A_moment_bounds}, the
conditional quadratic-form mean and variance then imply uniformly on bounded
$t$-intervals
\begin{equation}
\mathbb E|\mathbf a_i^{\mathsf T}\mathbf J_i(t)\mathbf a_i-h_i(t)|
=O(m^{-1/2}).
\end{equation}
Fubini and Markov therefore give
\begin{equation}
\int_{-B}^{B}|\mathbf a_i^{\mathsf T}\mathbf J_i(t)\mathbf a_i-h_i(t)|dt
\xrightarrow{\mathrm P}0.
\label{eq:B_integrated_concentration}
\end{equation}

It remains to control the adaptive-response trace itself. On a smooth response segment,
$\|\dot{\boldsymbol\gamma}_{\mathcal S}\|_2=O_{\mathrm P}(1)$ and
$\dot{\mathbf Q}=-\mathbf Q\mathbf A_{\mathcal S}\operatorname{diag}(\dot{\boldsymbol\gamma}_{\mathcal S})\mathbf A_{\mathcal S}^{\mathsf T}\mathbf Q$ give
$\|\dot{\mathbf Q}\|_{\mathrm F}=O_{\mathrm P}(1)$ because
$\|\operatorname{diag}(\dot{\boldsymbol\gamma}_{\mathcal S})\|_{\mathrm F}=\|\dot{\boldsymbol\gamma}_{\mathcal S}\|_2$. Hence
$\|\dot{\mathbf p}_{\mathcal S}\|_2+\|\dot{\mathbf G}_{\mathcal S}\|_{\mathrm F}+\|\dot{\mathbf H}_{\mathcal S}\|_{\mathrm F}+\|\dot{\mathbf L}_{\mathcal S}\|_{\mathrm F}=O_{\mathrm P}(1)$, and differentiating \eqref{eq:h_trace} gives $|\dot h_i(t)|=O_{\mathrm P}(m^{-1/2})$. At an entering crossing, let $\mathcal S_+=\{j\}\cup\mathcal C$,
$\mathbf c_j=\mathbf A_{\mathcal C}^{\mathsf T}\mathbf Q\mathbf a_j$,
$\boldsymbol\xi_j=2p_j\operatorname{diag}(\mathbf p_{\mathcal C})\mathbf c_j-\mathbf c_j^{\circ2}$,
$\mathbf B_{\mathcal C}=\mathbf Q\mathbf A_{\mathcal C}\operatorname{diag}(\mathbf p_{\mathcal C})$, and $\mathbf b_j=p_j\mathbf Q\mathbf a_j$.
With $\mathbf K:=\mathbf H_{\mathcal C}$ and $\mathbf J_-/\mathbf J_+$ the
pre/post-crossing Jacobians, block inversion gives
\begin{equation}
\mathbf J_+-\mathbf J_-=-2\kappa_j^{-1}\mathbf v_j\mathbf v_j^{\mathsf T},
\end{equation}
where $\mathbf v_j=\mathbf b_j-\mathbf B_{\mathcal C}\mathbf K^{-1}\boldsymbol\xi_j$ and
$\kappa_j=d_j^2-\boldsymbol\xi_j^{\mathsf T}\mathbf K^{-1}\boldsymbol\xi_j\ge c_H$; leaving is the reverse. By \eqref{eq:A_basic_bounds}, the Hessian gap, and uniform Gaussian spectral/column bounds, $\sup_j(\|\mathbf b_j\|+\|\mathbf B_{\mathcal C}\|_{\mathrm{op}}+\|\boldsymbol\xi_j\|_2)=O_{\mathrm P}(1)$ over regular crossings. Hence
$\sup_j|m^{-1}\operatorname{tr}(\mathbf J_+-\mathbf J_-)|=O_{\mathrm P}(m^{-1})$.
Assumption~\ref{ass:continuation} then gives
\begin{equation}
\sup_{|t|\le B}|h_i(t)-h_i(0)|=o_{\mathrm P}(1).
\label{eq:B_hstable}
\end{equation}
For the random endpoint $e_I=x_I-\mu_I=O_{\mathrm P}(1)$, truncate to $|e_I|\le B$,
integrate \eqref{eq:B_response_derivative}, and use
\eqref{eq:B_integrated_concentration}--\eqref{eq:B_hstable} together with
\eqref{eq:B_projection_transfer}; then
\begin{equation}
p_I=Z_{I,-}+h_{I,-}(x_I-\mu_I)+o_{\mathrm P}(1),
\end{equation}
which is \eqref{eq:quenched_scalar}.

The same smooth-path estimates give $\|\dot{\mathbf r}\|=O_{\mathrm P}(1)$ on both
response and insertion continuations. Since $\|\mathbf r\|=O_{\mathrm P}(\sqrt m)$,
$|\dot V|=O_{\mathrm P}(m^{-1/2})$; differentiating $m^{-1}\operatorname{tr}\mathbf Q$
gives the same order for $\dot\eta$. For $h$, the fixed-$\gamma_i$ insertion
has smooth derivative $O_{\mathrm P}(m^{-1/2})$ and $O_{\mathrm P}(m^{-1})$ normalized jumps;
allowing the final coordinate $\gamma_i$ to respond changes the normalized
trace by one Schur correction of order $O_{\mathrm P}(m^{-1})$. Hence
\begin{equation}
\begin{aligned}
\eta_{I,-}-\eta_n&=o_{\mathrm P}(1),&
V_{I,-}-V_n&=o_{\mathrm P}(1),\\
h_{I,-}-h_n&=o_{\mathrm P}(1),&&
\end{aligned}
\end{equation}
which proves \eqref{eq:cavity_macro_replace}.

\subsection{Two-Column Cavity and Empirical Scalar-Law Convergence}

For uniformly sampled distinct indices $I$ and $J$, let
$\mathbf y_{-IJ}:=\mathbf y-\mathbf a_Ix_I-\mathbf a_Jx_J$ and
$\mathbf r_{-IJ},\mathbf J_{-IJ}$ be the double-reduced residual/Jacobian. Put
$V_{-IJ}:=m^{-1}\|\mathbf r_{-IJ}\|^2$,
$h_{-IJ}:=m^{-1}\operatorname{tr}\mathbf J_{-IJ}$ and
$Z_k^{(2)}:=\mathbf a_k^{\mathsf T}\mathbf r_{-IJ}$ for $k=I,J$. With
$\mathcal F_{-IJ}:=\sigma(I,J,\mathbf A_{-IJ},\mathbf y_{-IJ},\zeta)$,
the deleted columns are conditionally independent Gaussian vectors, so
\begin{equation}
(Z_I^{(2)},Z_J^{(2)})
=\sqrt{V_{-IJ}}(G_I,G_J),
\quad G_I\perp\!\!\!\perp G_J\mid\mathcal F_{-IJ}.
\end{equation}
The only subtle point is to compare this double cavity with the one-coordinate
cavities without reintroducing dependence on the other deleted column. Keep
$I$ deleted and insert $J$ along the \emph{nested reduced insertion} whose
endpoint is the selected $(\mathbf A_{-I},\mathbf y_{-I})$ cavity and whose
inserted parameter $\gamma_{J\mid I,-}$ (the selected $J$th variance hyperparameter with $I$ deleted). The path is
$\mathcal F_{-IJ}\vee\sigma(\mathbf a_J)$-measurable and therefore remains
independent of $\mathbf a_I$. Its residual derivative has $O_{\mathrm P}(1)$ norm, so
$\|\mathbf r_{I,-}-\mathbf r_{-IJ}\|=O_{\mathrm P}(1)$. Since this displacement is
$\mathcal F_{-IJ}\vee\sigma(\mathbf a_J)$-measurable while $\mathbf a_I$ remains conditionally independent,
its conditional projected second moment equals
$m^{-1}\|\mathbf r_{I,-}-\mathbf r_{-IJ}\|^2=O_{\mathrm P}(m^{-1})$. Consequently, $Z_{I,-}-Z_I^{(2)}=o_{\mathrm P}(1)$. Interchanging $I$ and $J$ gives the
same conclusion for $J$. The corresponding large-system state insertion bounds give
\begin{equation}
\begin{aligned}
h_{I,-}-h_{-IJ}&=o_{\mathrm P}(1),& h_{J,-}-h_{-IJ}&=o_{\mathrm P}(1),\\
V_{I,-}-V_{-IJ}&=o_{\mathrm P}(1),& V_{J,-}-V_{-IJ}&=o_{\mathrm P}(1).
\end{aligned}
\end{equation}
Two successive insertions give $h_{-IJ}-h_n=o_{\mathrm P}(1)$ and
$V_{-IJ}-V_n=o_{\mathrm P}(1)$. Hence
\begin{align}
U_I&=h_{-IJ}x_I+\sqrt{V_{-IJ}}G_I+o_{\mathrm P}(1),
\label{eq:B_double_scalar_I}\\
U_J&=h_{-IJ}x_J+\sqrt{V_{-IJ}}G_J+o_{\mathrm P}(1),\\
h_{-IJ}-h_n&=o_{\mathrm P}(1),\qquad V_{-IJ}-V_n=o_{\mathrm P}(1).
\label{eq:B_double_macro}
\end{align}

For bounded Lipschitz $\varphi$, set
$D_i=\varphi(x_i,U_i)-\mathbb E_G\varphi(x_i,h_nx_i+\sqrt{V_n}G)$. Define
$\widetilde D_I:=\varphi(x_I,h_{-IJ}x_I+\sqrt{V_{-IJ}}G_I)-\mathbb E_G\varphi(x_I,h_{-IJ}x_I+\sqrt{V_{-IJ}}G)$, and similarly $\widetilde D_J$. Since $(h_{-IJ},V_{-IJ})$ is $\mathcal F_{-IJ}$-measurable and $G_I\perp G_J\mid\mathcal F_{-IJ}$, these variables are conditionally centered and independent. Equations \eqref{eq:B_double_scalar_I}--\eqref{eq:B_double_macro} and Lipschitzness give $D_I=\widetilde D_I+r_I$, $D_J=\widetilde D_J+r_J$ with $r_I,r_J=o_{\mathrm P}(1)$; boundedness gives $\mathbb E|r_I|+\mathbb E|r_J|\to0$, hence $\mathbb E D_I=o(1)$ and $\mathbb E[D_ID_J]\to0$. Therefore
\begin{equation}
\mathbb E\left[\left(n^{-1}\sum_iD_i\right)^2\right]
=n^{-1}\mathbb E D_I^2+(1-n^{-1})\mathbb E[D_ID_J]\to0,
\end{equation}
which proves the empirical scalar-law convergence in \eqref{eq:empirical_scalarization}. Tightness of the marked
empirical measures follows from \eqref{eq:A_basic_bounds}, bounded
$\gamma_i$, and $h_n=O_{\mathrm P}(1)$.

\section{Scalar KKT Geometry and Large-System Balances}
\label{app:macro}

This appendix proves the leave-one-out diagonal resolvent relation and the two balances in
Theorem~\ref{thm:macro_balances}.

\subsection{Adaptive Leave-One-Out Diagonal Resolvent Relation}

For comparison with the selected full stationary point, define the algebraic
full-nuisance leave-one-out quantities
\begin{equation}
\begin{aligned}
    \mathbf R_i^{\mathrm F}
    &=\left(
      \tau\mathbf I_m
      +\mathbf A_{-i}\boldsymbol{\Gamma}_{-i}^{\mathrm F}
       \mathbf A_{-i}^{\mathsf T}
    \right)^{-1},\\
    s_i^{\mathrm F}&=\mathbf a_i^{\mathsf T}\mathbf R_i^{\mathrm F}\mathbf a_i.
\end{aligned}
    \label{eq:algebraic_loo}
\end{equation}
The identity
\begin{equation}
    d_i=\frac{s_i^{\mathrm F}}{1+\gamma_i s_i^{\mathrm F}}.
    \label{eq:di_s}
\end{equation}
This follows from Sherman--Morrison.

Along the genuine reduced response path define
$s_{i,-}(t)=\mathbf a_i^{\mathsf T}\mathbf Q_{i,-}(t)\mathbf a_i$ and
$\eta_{i,-}(t)=m^{-1}\operatorname{tr}\mathbf Q_{i,-}(t)$. At $t=0$, conditional independence gives
\begin{equation}
\begin{aligned}
\mathbb E[s_{i,-}(0)\mid\mathcal F_{-i}]&=\eta_{i,-}(0),\\
\operatorname{Var}(s_{i,-}(0)\mid\mathcal F_{-i})&=O_{\mathrm P}(m^{-1}).
\end{aligned}
\end{equation}
Writing $\mathbf c_{i,-}=\mathbf A_{-i}^{\mathsf T}\mathbf Q_{i,-}\mathbf a_i$,
$q_j=m^{-1}\mathbf a_j^{\mathsf T}\mathbf Q_{i,-}^2\mathbf a_j$, and
$\mathbf q=(q_j)_{j\ne i}$,
\begin{equation}
\frac{d}{dt}(s_{i,-}-\eta_{i,-})
=-\sum_{j\ne i}\dot\gamma_j(c_{ji,-}^2-q_j).
\end{equation}
Appendix~\ref{app:scalar} gives
$\|\dot{\boldsymbol\gamma}\|_2=O_{\mathrm P}(1)$ and
$\|\mathbf c_{i,-}^{\circ2}\|_2=O_{\mathrm P}(\rho_n)$, while
$\|\mathbf q\|_2=O_{\mathrm P}(m^{-1/2})$. Hence, uniformly on bounded response
intervals,
\begin{equation}
\sup_{|t|\le B}|s_{i,-}(t)-\eta_{i,-}(t)|=O_{\mathrm P}(\rho_n).
\end{equation}
At the pseudo-response endpoint this is $o_{\mathrm P}(1)$. The resolvent identity and
\eqref{eq:B_gamma_transfer} then transfer the relation to the algebraic
full-nuisance resolvent $\mathbf R_i^{\mathrm F}$; reinserting the single covariance
term changes the normalized trace by $O_{\mathrm P}(m^{-1})$. Thus
\begin{equation}
    s_I^{\mathrm F}-\eta_n=o_{\mathrm P}(1),
    \label{eq:selfscreen_s}
\end{equation}
and the $1$-Lipschitz map $s\mapsto s/(1+\gamma s)$, together with
\eqref{eq:di_s}, gives \eqref{eq:di_selfscreen}.
Combining this with finite-dimensional KKT inequalities gives the limiting
inactive and active scalar graph \eqref{eq:scalar_inactive}--\eqref{eq:scalar_active},
including active sequences that approach the boundary $\Gamma=0$.

\subsection{Large-System Balances}

From $\mathbf C\mathbf Q=\mathbf I_m$, we have
\begin{equation}
m=\tau\operatorname{tr}\mathbf Q+\sum_i\gamma_i d_i.
\end{equation}
At a KKT point $\gamma_i d_i=\mu_i p_i$, giving
\eqref{eq:exact_trace_balance}. Since
$|\mu_i p_i|\le\gamma_{\max}\tau_\star^{-1}\|\mathbf a_i\|^2$, Gaussian
column moments provide uniform integrability, so weak convergence of the
marked empirical law and the scalar KKT relation give \eqref{eq:balance_eta}.

For the noise balance we use the following second-order Gaussian
Stein identity. If $\mathbf z\sim\mathcal N(0,\sigma^2\mathbf I_m)$ and
$\mathbf f\in W^{1,2}$ has symmetric weak Jacobian $\mathbf J_f$, then
\begin{equation}
\mathbb E[(\mathbf z^{\mathsf T}\mathbf f-
\sigma^2\operatorname{tr}\mathbf J_f)^2]
=\sigma^2\mathbb E\|\mathbf f\|^2+
\sigma^4\mathbb E\|\mathbf J_f\|_{\mathrm F}^2.
\label{eq:C_second_order_stein}
\end{equation}
This is the second-order Gaussian Stein identity; its weakly differentiable
form follows by Sobolev approximation \cite{Stein1981,BellecZhang2021}. Conditional on $(\mathbf A,\zeta)$,
Lemma~\ref{lem:sobolev_residual} applies it to the complete adaptive map
$\mathbf v\mapsto\mathbf r(\mathbf A\mathbf x+\mathbf v)$. With
\eqref{eq:A_moment_bounds},
\begin{equation}
\mathbb E\left[\left(m^{-1}\mathbf v^{\mathsf T}\mathbf r-
\sigma_0^2h_n\right)^2\right]=O(m^{-1}),
\end{equation}
which proves \eqref{eq:stein_noise} without any independence assumption
between $\mathbf v$ and $\mathbf r$. Multiplying \eqref{eq:noise_identity} by
$\mathbf r^{\mathsf T}/m$ gives
\begin{equation}
m^{-1}\mathbf v^{\mathsf T}\mathbf r
=\tau V_n+m^{-1}\sum_i(\mu_i-x_i)p_i.
\end{equation}
The same Gaussian moment bounds give uniform integrability of
$(\mu_I-x_I)p_I$, and passage to the limiting marked law yields
\eqref{eq:balance_noise}.

\section{Adaptive-Response Closure}
\label{app:susceptibility}

This appendix proves Theorem~\ref{thm:susceptibility}. We first derive the
finite-dimensional profiled derivative and then identify its large-system
adaptive response through a random-endpoint response-trace argument.

\subsection{Finite-Dimensional Profiled Coordinate Response}

Fix an active coordinate $i$ and write $\mathcal S=\{i\}\cup\mathcal C$.
With
\begin{equation}
\mathbf G_{\mathcal S}=
\begin{bmatrix}d_i&\mathbf c_i^{\mathsf T}\\
\mathbf c_i&\mathbf G_{\mathcal C}\end{bmatrix},\quad
\mathbf w_i=\operatorname{diag}(\mathbf p_{\mathcal C})\mathbf c_i,\quad
\mathbf z_i=\mathbf c_i^{\circ2},
\end{equation}
set
$\mathbf K_i=2\operatorname{diag}(\mathbf p_{\mathcal C})\mathbf G_{\mathcal C}
\operatorname{diag}(\mathbf p_{\mathcal C})-\mathbf G_{\mathcal C}\circ\mathbf G_{\mathcal C}$.
Since $p_i^2=d_i$,
\begin{equation}
\mathbf H_{\mathcal S}=
\begin{bmatrix}d_i^2&\boldsymbol{\xi}_i^{\mathsf T}\\
\boldsymbol{\xi}_i&\mathbf K_i\end{bmatrix},\qquad
\boldsymbol{\xi}_i=2p_i\mathbf w_i-\mathbf z_i.
\end{equation}
Define
$\vartheta_i=d_i-2\mathbf w_i^{\mathsf T}\mathbf K_i^{-1}\mathbf w_i$,
$\psi_i=\mathbf z_i^{\mathsf T}\mathbf K_i^{-1}\mathbf w_i$, and
$\phi_i=\mathbf z_i^{\mathsf T}\mathbf K_i^{-1}\mathbf z_i$.
The Schur complement of $\mathbf K_i$ is
$d_i^2-\boldsymbol{\xi}_i^{\mathsf T}\mathbf K_i^{-1}\boldsymbol{\xi}_i$.
We denote this profiled active-coordinate KKT curvature by $\kappa_i$; at a
regular support transition it is the same one-dimensional profiled curvature
used in Appendix~B. Substituting the
definitions above gives
\begin{equation}
    \kappa_i=d_i(2\vartheta_i-d_i)+4p_i\psi_i-\phi_i.
    \label{eq:D_kappa}
\end{equation}
Perturbing $\mathbf y(t)=\mathbf y+t\mathbf a_i$ and differentiating gives
\begin{equation}
\mathbf H_{\mathcal S}
\begin{bmatrix}\dot\gamma_i\\ \dot{\boldsymbol\gamma}_{\mathcal C}\end{bmatrix}
=\begin{bmatrix}2p_id_i\\2\mathbf w_i\end{bmatrix},\qquad
\dot\gamma_i=\frac{2p_i\vartheta_i+2\psi_i}{\kappa_i},
\end{equation}
and $\dot p_i=\vartheta_i-(p_i\vartheta_i+\psi_i)\dot\gamma_i$.
Since $\mu_i=\gamma_ip_i$,
\begin{equation}
\begin{aligned}
    \dot\mu_i
    =\gamma_i\vartheta_i+[p_i(1-\gamma_i\vartheta_i)-\gamma_i\psi_i]\dot\gamma_i.
\end{aligned}
    \label{eq:D_mudot}
\end{equation}
The identities above are finite-dimensional. From this point onward, the
$O_{\mathrm P}$ and $o_{\mathrm P}$ statements in the coordinatewise asymptotics refer to a uniformly
sampled index $I$, and active-coordinate statements are understood on
$\{\gamma_I>0\}$. The nuisance block $\mathbf K_I$ is a principal submatrix
of $\mathbf H_{\mathcal S}$, hence
$\mathbf K_I\succeq c_H\mathbf I$. Appendix~B gives
$\|\mathbf c_I\|_\infty=O_{\mathrm P}(\rho_n)$ and
$\|\mathbf c_I\|_2=O_{\mathrm P}(1)$. Together with the uniform bound on
$\|\mathbf D_{\mathbf p}\|_{\mathrm{op}}$ this gives
$\psi_I=O_{\mathrm P}(\rho_n)$ and $\phi_I=O_{\mathrm P}(\rho_n^2)$.
Moreover, the scalar Schur complement satisfies
$\kappa_I^{-1}=(\mathbf H_{\mathcal S}^{-1})_{II}$, so
$\kappa_I\ge c_H$. Equation~\eqref{eq:D_kappa} therefore implies
\begin{equation}
    d_I(2\vartheta_I-d_I)\ge c_H/2,
    \label{eq:D_exact_denominator_gap}
\end{equation}
with probability tending to one on $\{\gamma_I>0\}$. Since
$d_I\le \tau_\star^{-1}\|\mathbf a_I\|^2=O_{\mathrm P}(1)$, the exact denominator
$2\vartheta_I-d_I$ is separated from zero in probability. We may therefore
substitute the estimates for $\psi_I$ and $\phi_I$ into
\eqref{eq:D_kappa}--\eqref{eq:D_mudot} without dividing by a vanishing
quantity. Therefore, for a uniformly sampled active index,
\begin{equation}
    \mathbf a_I^{\mathsf T}\nabla_{\mathbf y}\mu_I
    =\frac{\vartheta_I(2-\gamma_I d_I)}{2\vartheta_I-d_I}+o_{\mathrm P}(1).
    \label{eq:mu_directional}
\end{equation}

\subsection{Adaptive-Response Trace Transfer}

Remove the $i$th covariance term while retaining the full nuisance state:
\begin{equation}
\mathbf R_i^{\mathrm F}=(\tau\mathbf I_m+\mathbf A_{\mathcal C}
\boldsymbol\Gamma_{\mathcal C}^{\mathrm F}\mathbf A_{\mathcal C}^{\mathsf T})^{-1},
\qquad
s_i^{\mathrm F}=\mathbf a_i^{\mathsf T}\mathbf R_i^{\mathrm F}\mathbf a_i,
\end{equation}
with $\mathbf c_i^{\mathrm R}=\mathbf A_{\mathcal C}^{\mathsf T}\mathbf R_i^{\mathrm F}\mathbf a_i$.
At the algebraic pseudo-response point,
$\mathbf r^{\mathrm F}=\mathbf R_i^{\mathrm F}\mathbf u_i^{\mathrm{ps}}$ and
$\mathbf p_{\mathcal C}^{\mathrm R}=\mathbf p_{\mathcal C}^{\mathrm F}$. With
$\mathbf w_i^{\mathrm R}=\operatorname{diag}(\mathbf p_{\mathcal C}^{\mathrm R})\mathbf c_i^{\mathrm R}$ and
$\mathbf G_i^{\mathrm R}=\mathbf A_{\mathcal C}^{\mathsf T}\mathbf R_i^{\mathrm F}\mathbf A_{\mathcal C}$,
define
\begin{equation}
\mathbf K_i^{\mathrm R}=2\operatorname{diag}(\mathbf p_{\mathcal C}^{\mathrm R})\mathbf G_i^{\mathrm R}
\operatorname{diag}(\mathbf p_{\mathcal C}^{\mathrm R})-\mathbf G_i^{\mathrm R}\circ\mathbf G_i^{\mathrm R},
\end{equation}
and
$\ell_i^{\mathrm R}=s_i^{\mathrm F}-2(\mathbf w_i^{\mathrm R})^{\mathsf T}(\mathbf K_i^{\mathrm R})^{-1}\mathbf w_i^{\mathrm R}$.
We call $\ell_i^{\mathrm R}$ an algebraic profiled response rather than a Jacobian,
because the full nuisance vector is only an approximate KKT solution of the
reduced pseudo-response problem.

Let $e_i:=x_i-\mu_i^{\mathrm F}$ and define
$\ell_i^{\mathrm{ps}}:=\mathbf a_i^{\mathsf T}\mathbf J_i^{\mathrm{ps}}\mathbf a_i$,
where $\mathbf J_i^{\mathrm{ps}}$ is the genuine reduced Jacobian at $t=e_i$.
By \eqref{eq:B_gamma_transfer}, the algebraic and genuine pseudo-response
nuisance vectors are $O_{\mathrm P}(\rho_n)$ apart. On the no-crossing event in the
second clause of Assumption~\ref{ass:derivative}, the correction remains in a
single smooth nuisance branch segment, so its two endpoints belong to the same nuisance
active face. Interpolate linearly between the two nuisance hyperparameter
vectors on this fixed face while keeping $\mathbf u_i^{\mathrm{ps}}$ fixed.
Let $\Delta\boldsymbol\gamma_{\mathcal C}:=
\boldsymbol\gamma_{\mathcal C}^{\mathrm F}-\boldsymbol\gamma_{\mathcal C}^{\mathrm{ps}}$ and
$\boldsymbol\gamma_{\mathcal C}(u):=\boldsymbol\gamma_{\mathcal C}^{\mathrm{ps}}
+u\Delta\boldsymbol\gamma_{\mathcal C}$, $u\in[0,1]$; then
$\|\Delta\boldsymbol\gamma_{\mathcal C}\|_2=O_{\mathrm P}(\rho_n)$. Here $s_i,\mathbf c_i,\mathbf G_i,\mathbf w_i,\mathbf K_i$ denote their current-resolvent counterparts.

Standard resolvent differentiation gives
\begin{equation}
    \dot{\mathbf R}
    =-\mathbf R\mathbf A_{\mathcal C}
      \operatorname{diag}(\Delta\boldsymbol\gamma_{\mathcal C})
      \mathbf A_{\mathcal C}^{\mathsf T}\mathbf R.
\end{equation}
Starting from the genuine pseudo-response endpoint, where Appendix~B gives
$\|\mathbf c_i\|_\infty=O_{\mathrm P}(\rho_n)$ and
$\|\mathbf c_i\|_2=O_{\mathrm P}(1)$, Gronwall's inequality yields the same bounds
uniformly along the interpolation. The nuisance residual projections satisfy
$\dot{\mathbf p}_{\mathcal C}
=-\mathbf G_i(\Delta\boldsymbol\gamma_{\mathcal C}\circ\mathbf p_{\mathcal C})$;
since their maximum norm is $O_{\mathrm P}(1)$ at the genuine stationary endpoint,
another one-dimensional Gronwall bound preserves
$\|\mathbf p_{\mathcal C}\|_\infty=O_{\mathrm P}(1)$ on the comparison segment.
Consequently,
$|\dot s_i|=O_{\mathrm P}(\rho_n^2)$,
$\|\dot{\mathbf p}_{\mathcal C}\|_\infty=O_{\mathrm P}(\rho_n)$,
$\|\dot{\mathbf c}_i\|_2=O_{\mathrm P}(\rho_n^2)$, and
$\|\dot{\mathbf w}_i\|_2=O_{\mathrm P}(\rho_n)$.
Moreover,
$\|\dot{\mathbf G}_i\|_{\mathrm F}=O_{\mathrm P}(\rho_n)$. Since the entries of
$\mathbf G_i$ are uniformly $O_{\mathrm P}(1)$, the Hadamard-square term obeys
\begin{equation}
    \left\|\frac{d}{du}(\mathbf G_i\circ\mathbf G_i)\right\|_{\mathrm{op}}
    \leq
    \left\|\frac{d}{du}(\mathbf G_i\circ\mathbf G_i)\right\|_{\mathrm F}
    =O_{\mathrm P}(\rho_n),
\end{equation}
so $\|\dot{\mathbf K}_i\|_{\mathrm{op}}=O_{\mathrm P}(\rho_n)$.
At the genuine pseudo-response stationary point, the nuisance Hessian is a
principal active block and is therefore bounded below by $c_H\mathbf I$.
The perturbation above keeps $\mathbf K_i\succeq(c_H/2)\mathbf I$ uniformly
with probability tending to one. Differentiating
$\ell_i=s_i-2\mathbf w_i^{\mathsf T}\mathbf K_i^{-1}\mathbf w_i$ then gives
$|\dot\ell_i|=O_{\mathrm P}(\rho_n)$, and hence
\begin{equation}
    \ell_i^{\mathrm R}-\ell_i^{\mathrm{ps}}=O_{\mathrm P}(\rho_n)=o_{\mathrm P}(1).
    \label{eq:D_reaction_transfer}
\end{equation}
\subsection{Random-Endpoint Adaptive-Response Concentration}

Let $\mathbf J_i(t)$ be the nuisance-adaptive residual Jacobian of the
genuine reduced branch already defined in Appendix~B.
For deterministic bounded $t$, Gaussian regression gives
\begin{equation}
    \mathbf a_i^{\mathsf T}\mathbf J_i(t)\mathbf a_i
    -\frac1m\operatorname{tr}\mathbf J_i(t)
    =o_{\mathrm P}(1).
    \label{eq:D_fixed_t}
\end{equation}
The endpoint $e_I=x_I-\mu_I$ is random for a uniformly sampled coordinate.
We first record that the derivative envelope needed for interpolation follows
from Assumptions~\ref{ass:model}--\ref{ass:regular}, rather than being an
additional assumption. On a smooth reduced response segment,
\begin{equation}
    \dot{\boldsymbol\gamma}_{\mathcal S}
    =2\mathbf H_{\mathcal S}^{-1}\mathbf D_{\mathbf p}
      \mathbf A_{\mathcal S}^{\mathsf T}\mathbf Q\mathbf a_I,
\end{equation}
so the Hessian gap, the resolvent bound, and the Gaussian spectral- and
column-norm bounds give $\|\dot{\boldsymbol\gamma}_{\mathcal S}\|_2=O_{\mathrm P}(1)$
uniformly over the smooth segments of any fixed bounded response interval.
The same diagonal-perturbation argument as in Appendix~B gives
$\|\dot{\mathbf Q}\|_{\mathrm F}+\|\dot{\mathbf p}_{\mathcal S}\|_2+\|\dot{\mathbf G}_{\mathcal S}\|_{\mathrm F}+\|\dot{\mathbf H}_{\mathcal S}\|_{\mathrm F}=O_{\mathrm P}(1)$. Differentiating \eqref{eq:J_formula} therefore yields
\begin{equation}
    \sup_{t\in[-B,B]\,\mathrm{smooth}}
    \left\|\frac{d}{dt}\mathbf J_I(t)\right\|_{\mathrm{op}}
    =O_{\mathrm P}(1).
    \label{eq:D_J_derivative_derived}
\end{equation}
Thus the interpolation regularity required below is a consequence of the
regular branch bounds rather than an additional assumption.

For every $\xi>0$, first choose a fixed $B$ such that
$\limsup_n\Pr(|e_I|>B)<\xi$. On $|e_I|\leq B$, choose only in the proof a deterministic grid
$\mathcal G_n(B)\subset[-B,B]$ of mesh
$\varepsilon_n=\min\{m^{-1},\omega_n^2\}$.
Then $\varepsilon_n=o(\omega_n)$,
$\varepsilon_n\to0$, and
$\log|\mathcal G_n(B)|=o(m)$ by Assumption~\ref{ass:derivative}.
Let
\begin{equation}
\Delta_i^{\mathrm{quad}}(t)
:=\mathbf a_i^{\mathsf T}\mathbf J_i(t)\mathbf a_i
-m^{-1}\operatorname{tr}\mathbf J_i(t).
\end{equation}
For deterministic $t$, condition on the reduced response and selection randomization.
The Gaussian regression representation of Appendix~B writes
$\mathbf a_i=\mathbf m_i(t)+\varsigma_i(t)\,\mathbf z$ with
$\mathbf z\sim\mathcal N(0,\mathbf I_m)$,
$\varsigma_i^2(t)=m^{-1}+O(m^{-2})$, and
$\|\mathbf m_i(t)\|=O_{\mathrm P}(m^{-1/2})$. Fix a sufficiently large constant $M$. The operator bound in Appendix~A is
controlled by $\|\mathbf A_{-i}\|_{\mathrm{op}}$, the maximum reduced-column
norm, $\tau_\star$, and $c_H$, and is therefore uniform in $t$. Likewise,
$\sup_{|t|\leq B}\|\mathbf m_i(t)\|\leq C m^{-1/2}$ on a Gaussian norm
event whose probability tends to one. Let $\mathcal E_n^{\mathrm{glob}}(B,M)$
denote such a common high-probability envelope. For each deterministic grid
point $t$, also let $\mathcal E_{n,t}(M)$ be the event
\begin{align}
  \|\mathbf J_i(t)\|_{\mathrm{op}}&\leq M,
  &\|\mathbf J_i(t)\|_{\mathrm F}^2&\leq Mm,\\
  \|\mathbf m_i(t)\|&\leq Mm^{-1/2}.&&
\end{align}
The event $\mathcal E_{n,t}(M)$ is $\mathcal G_i(t)$-measurable, while
$\mathcal E_n^{\mathrm{glob}}(B,M)\subseteq
 \bigcap_{t\in\mathcal G_n(B)}\mathcal E_{n,t}(M)$ and
$\Pr((\mathcal E_n^{\mathrm{glob}})^c)=o(1)$ after $M$ is chosen large.
The concentration is conditioned separately at each deterministic grid point,
not on the collection $\{\mathbf u_i(t):t\in\mathcal G_n(B)\}$, which
would reveal additional information about $\mathbf a_i$.

For each fixed grid point, conditional Gaussian quadratic-form concentration
\cite{RudelsonVershynin2013,Vershynin2018} on the
$\mathcal G_i(t)$-measurable event $\mathcal E_{n,t}(M)$ gives, for
fixed $\epsilon_0>0$ and constants $C_M,c_M>0$ depending only on $M$,
\begin{equation}
  \Pr\!\left(
    |\Delta_i^{\mathrm{quad}}(t)|>\epsilon_0,
    \mathcal E_{n,t}(M)
  \right)
  \leq C_Me^{-c_Mm\epsilon_0^2},
\end{equation}
where the $O(m^{-1})$ conditional bias is absorbed into the threshold.
Removing the conditioning and taking a union bound gives
\begin{equation}
\begin{aligned}
\Pr\!\left(\max_{t\in\mathcal G_n(B)}
|\Delta_i^{\mathrm{quad}}(t)|>\epsilon_0\right)
&\le o(1)+|\mathcal G_n(B)|C_Me^{-c_Mm\epsilon_0^2}\\
&\to0.
\end{aligned}
\end{equation}
Let $t_n$ be the nearest grid point to $e_I$. By the first clause of
Assumption~\ref{ass:derivative}, $e_I$ is farther than $\omega_n$ from the
response-transition set with probability tending to one. Since
$\varepsilon_n=o(\omega_n)$, $e_I$ and $t_n$ then lie on the same smooth
branch segment. Equation~\eqref{eq:D_J_derivative_derived} and
$\varepsilon_n\to0$ give
\begin{equation}
    \|\mathbf J_I(e_I)-\mathbf J_I(t_n)\|_{\mathrm{op}}=o_{\mathrm P}(1).
\end{equation}
Since $\|\mathbf a_I\|^2=O_{\mathrm P}(1)$, the quadratic-form difference is also
$o_{\mathrm P}(1)$, while the normalized trace difference is bounded by the same
operator norm. Therefore
\begin{equation}
    \mathbf a_I^{\mathsf T}\mathbf J_I(e_I)\mathbf a_I
    -\frac1m\operatorname{tr}\mathbf J_I(e_I)
    =o_{\mathrm P}(1).
    \label{eq:D_random_endpoint_quad}
\end{equation}
The adaptive-response trace stability from Appendix~B gives
\begin{equation}
    \frac1m\operatorname{tr}\mathbf J_I(e_I)
    =h_n+o_{\mathrm P}(1).
\end{equation}
Consequently
$\ell_I^{\mathrm{ps}}=h_n+o_{\mathrm P}(1)$. On
$\{\gamma_I^{\mathrm F}>0\}$, the no-crossing clause of
Assumption~\ref{ass:derivative} and \eqref{eq:D_reaction_transfer} then give
\begin{equation}
    \ell_I^{\mathrm R}=h_n+o_{\mathrm P}(1).
    \label{eq:D_ell_selfscreen}
\end{equation}
Letting $\xi\downarrow0$ removes the truncation.

Let $\beta_i=(1+\gamma_i s_i^{\mathrm F})^{-1}$ and
$\lambda_i=\gamma_i\beta_i$. Sherman--Morrison gives
$d_i=\beta_i s_i^{\mathrm F}$, $\mathbf c_i=\beta_i\mathbf c_i^{\mathrm R}$, and
$\mathbf w_i=\beta_i\mathbf w_i^{\mathrm R}$.
Moreover,
\begin{equation}
\begin{aligned}
    \mathbf K_i
    ={}&\mathbf K_i^{\mathrm R}
    -2\lambda_i\mathbf w_i^{\mathrm R}(\mathbf w_i^{\mathrm R})^{\mathsf T}
    +\mathbf E_i^{\mathrm K},
\end{aligned}
\end{equation}
where, with $\mathbf z_i^{\mathrm R}=(\mathbf c_i^{\mathrm R})^{\circ2}$,
\begin{equation}
    \mathbf E_i^{\mathrm K}
    =2\lambda_i\operatorname{diag}(\mathbf c_i^{\mathrm R})
      \mathbf G_i^{\mathrm R}\operatorname{diag}(\mathbf c_i^{\mathrm R})
     -\lambda_i^2\mathbf z_i^{\mathrm R}(\mathbf z_i^{\mathrm R})^{\mathsf T}.
\end{equation}
Since $\|\mathbf c_i^{\mathrm R}\|_\infty=O_{\mathrm P}(\rho_n)$,
$\|\mathbf G_i^{\mathrm R}\|_{\mathrm{op}}=O_{\mathrm P}(1)$, and
$\|\mathbf z_i^{\mathrm R}\|_2=O_{\mathrm P}(\rho_n)$,
$\|\mathbf E_i^{\mathrm K}\|_{\mathrm{op}}=O_{\mathrm P}(\rho_n^2)=o_{\mathrm P}(1)$.
Let
$\mathbf K_i^0:=\mathbf K_i^{\mathrm R}-2\lambda_i\mathbf w_i^{\mathrm R}
(\mathbf w_i^{\mathrm R})^{\mathsf T}$. Since the full nuisance block
$\mathbf K_i$ is a principal submatrix of
$\mathbf H_{\mathcal S}\succeq c_H\mathbf I$ and
$\|\mathbf E_i^{\mathrm K}\|_{\mathrm{op}}=o_{\mathrm P}(1)$,
$\mathbf K_i^0\succeq(c_H/2)\mathbf I$ with probability tending to one.
The Sherman--Morrison denominator for this rank-one update is
\begin{equation}
1-2\lambda_i(\mathbf w_i^{\mathrm R})^{\mathsf T}
(\mathbf K_i^{\mathrm R})^{-1}\mathbf w_i^{\mathrm R}
=\beta_i(1+\gamma_i\ell_i^{\mathrm R}).
\end{equation}
The spectral gaps of $\mathbf K_i^{\mathrm R}$ and $\mathbf K_i^0$, together with the
bounded operator norm of $\mathbf K_i^{\mathrm R}$, keep this denominator uniformly
away from zero. Indeed, since $\lambda_i\geq0$,
\begin{equation}
\begin{aligned}
    1-2\lambda_i(\mathbf w_i^{\mathrm R})^{\mathsf T}
      (\mathbf K_i^{\mathrm R})^{-1}\mathbf w_i^{\mathrm R}\geq
      \frac{\lambda_{\min}(\mathbf K_i^0)}
           {\|\mathbf K_i^{\mathrm R}\|_{\mathrm{op}}}.
\end{aligned}
\end{equation}
For a uniformly sampled coordinate $I$, the bound
$s_I^{\mathrm F}\leq\tau_\star^{-1}\|\mathbf a_I\|^2$ and bounded $\gamma_I$ imply
that $\beta_I$ is bounded away from zero in probability. Hence, on
$\{\gamma_I>0\}$, $1+\gamma_I\ell_I^{\mathrm R}$ is also bounded away from zero in
probability. Since
$\mathbf K_i^{-1}=(\mathbf K_i^0)^{-1}+o_{\mathrm P}(1)$ in operator norm, the
rank-one Sherman--Morrison identity gives
\begin{equation}
    (\mathbf w_i^{\mathrm R})^{\mathsf T}(\mathbf K_i^0)^{-1}\mathbf w_i^{\mathrm R}
    =\frac{(\mathbf w_i^{\mathrm R})^{\mathsf T}(\mathbf K_i^{\mathrm R})^{-1}\mathbf w_i^{\mathrm R}}
    {1-2\lambda_i(\mathbf w_i^{\mathrm R})^{\mathsf T}
      (\mathbf K_i^{\mathrm R})^{-1}\mathbf w_i^{\mathrm R}}.
\end{equation}
Substitution into
$\vartheta_i=d_i-2\beta_i^2(\mathbf w_i^{\mathrm R})^{\mathsf T}
\mathbf K_i^{-1}\mathbf w_i^{\mathrm R}$ and the identity defining $\ell_i^{\mathrm R}$
yield
\begin{equation}
    \vartheta_i
    =\frac{\ell_i^{\mathrm R}}{1+\gamma_i\ell_i^{\mathrm R}}+o_{\mathrm P}(1).
\end{equation}
Together with the adaptive-response transfer proved above, this yields
\begin{equation}
    \vartheta_I=\frac{h_n}{1+h_n\gamma_I}+o_{\mathrm P}(1).
    \label{eq:theta_selfscreen}
\end{equation}
Combining \eqref{eq:theta_selfscreen} with \eqref{eq:di_selfscreen} gives,
for a uniformly sampled coordinate $I$ on
$\{\gamma_I>0\}$,
\begin{equation}
    \mathbf a_I^{\mathsf T}\nabla_{\mathbf y}\mu_I
    =h_n\frac{2+\eta_n\gamma_I}
    {(2h_n-\eta_n)+\eta_nh_n\gamma_I}+o_{\mathrm P}(1).
    \label{eq:D_coord_susceptibility}
\end{equation}
A strictly inactive coordinate has zero local derivative.

The denominator in \eqref{eq:D_coord_susceptibility} is controlled by the
finite-dimensional Schur gap rather than by an assumption such as
$2h_n>\eta_n$. Equation~\eqref{eq:D_exact_denominator_gap} already shows
that $2\vartheta_I-d_I$ is separated from zero in probability on a typical
active coordinate. In addition, the stable rank-one denominator above and
\eqref{eq:D_ell_selfscreen} keep $1+h_n\gamma_I$ separated from zero, while
$1+\eta_n\gamma_I\geq1$. Therefore
\begin{equation}
2\frac{h_n}{1+h_n\gamma_i}
-\frac{\eta_n}{1+\eta_n\gamma_i}
=
\frac{(2h_n-\eta_n)+\eta_nh_n\gamma_i}
{(1+h_n\gamma_i)(1+\eta_n\gamma_i)},
\end{equation}
is also separated from zero with probability tending to one for a uniformly
sampled active coordinate. Consequently every active limiting mark satisfies
$ (2h-\eta)+\eta h\Gamma\neq0$, and the rational adaptive-response factor is
well-defined on the active part of the limiting marked law. No sign condition
on $2h_n-\eta_n$ is used.

There is one boundary issue to exclude before taking empirical limits. From
Section~\ref{sec:model}, every regular limiting state has $V>0$. Hence the
conditional scalar law $U\mid(X=x,\mathcal M)\sim\mathcal N(hx,V)$ has a
continuous density, so $\nu_b\{U^2=\eta\}=0$ almost surely. To track activity,
augment the marks by
$A_i^{\mathrm{act}}:=\mathbf 1_{\{\gamma_i>0\}}$.
Along any further subsequence this binary mark has a limiting mark
$A^{\mathrm{act}}\in\{0,1\}$. On $\{A^{\mathrm{act}}=1\}$, the active
finite-dimensional KKT relation together with the leave-one-out resolvent limits gives
$P^2=\eta/(1+\eta\Gamma)$ and $U=P(1+h\Gamma)$. Therefore
$\{A^{\mathrm{act}}=1,\Gamma=0\}$ is contained in $\{U^2=\eta\}$ and
has zero limiting mass. Conversely, $A^{\mathrm{act}}=0$ implies
$\Gamma=0$. Hence $A^{\mathrm{act}}=\mathbf 1_{\{\Gamma>0\}}$ almost surely under the
limiting marked law. This justifies the activity
indicator in Theorem~\ref{thm:susceptibility} without a boundary term.

For the empirical trace, first note that the exact directional response
has a uniform moment envelope. Along $\mathbf y+t\mathbf a_i$,
\begin{equation}
\|\dot{\boldsymbol\gamma}_{\mathcal S}\|_2
\leq2c_H^{-1}\|\mathbf D_{\mathbf p}\|_{\mathrm{op}}
\|\mathbf A\|_{\mathrm{op}}\|\mathbf Q\|_{\mathrm{op}}
\|\mathbf a_i\|,
\end{equation}
and
\begin{equation}
|\dot p_i|
\leq d_i+\|\mathbf A\|_{\mathrm{op}}\|\mathbf Q\|_{\mathrm{op}}
\|\mathbf a_i\|\|\mathbf D_{\mathbf p}\|_{\mathrm{op}}
\|\dot{\boldsymbol\gamma}_{\mathcal S}\|_2.
\end{equation}
Thus $|\dot\mu_i|$ is bounded by a fixed-degree polynomial in Gaussian
spectral and column norms. Consequently, for some $r>1$,
\begin{equation}
\sup_n\mathbb E
|\mathbf a_I^{\mathsf T}\nabla_{\mathbf y}\mu_I|^r<\infty.
\end{equation}
This supplies uniform integrability without imposing a separate moment
assumption on the scalar rational expression.

To pass from a typical coordinate to the empirical trace, augment the marked
measure by $A_i^{\mathrm{act}}$ and
$S_i^{\mathrm{dir}}:=\mathbf a_i^{\mathsf T}\nabla_{\mathbf y}\mu_i$, taking
$S_i^{\mathrm{dir}}=0$ inside an inactive face. The moment bound above gives
tightness and uniform integrability. By \eqref{eq:D_coord_susceptibility} and
the Schur-gap control, the limit satisfies
$S^{\mathrm{dir}}=h\Psi_{\eta,h}(\Gamma)$, with
$A^{\mathrm{act}}=\mathbf 1_{\{\Gamma>0\}}$ almost surely. Differentiating
\eqref{eq:reconstruction_identity} and taking normalized traces gives
$1=\tau h_n+(n/m)n^{-1}\sum_iS_i^{\mathrm{dir}}$; passage to the augmented
limit and $n/m\to1/\delta$ yield \eqref{eq:balance_h}.

\section{Joint Noise Variance Estimation and Reconstruction Error}
\label{app:noise}

This appendix proves the noise-sensitivity and MSE statements. Bars denote
fixed-$\tau$ hyperparameter-profiled quantities $(\overline{\mathbf r},
\overline{\boldsymbol\mu})$, with $\overline{\mathbf J}:=D_{\mathbf y}
\overline{\mathbf r}$ and $\overline h_n:=m^{-1}\operatorname{tr}\overline{\mathbf J}$.
The selected full joint value is $\widehat\tau$; at that point
$\mathbf J^{\mathrm{jt}}=\mathbf J_{\mathrm{joint}}$ and
$h_n^{\mathrm{joint}}:=m^{-1}\operatorname{tr}\mathbf J_{\mathrm{joint}}$. Unless a joint decoration is shown, local quantities such as $\mathbf Q$, $\mathbf p$, and $\mathbf H_{\mathcal S}$ in fixed-$\tau$ derivatives are evaluated at the same hyperparameter-profiled state.

\subsection{Noise-Variance Sensitivity}

Differentiating the active KKT equations at fixed $\tau$ gives
\begin{equation}
\begin{aligned}
\mathbf H_{\mathcal S}\partial_\tau\overline{\boldsymbol\gamma}_{\mathcal S}
&=-\partial_\tau\mathbf g_{\mathcal S},\\
\partial_\tau\mathbf g_{\mathcal S}
&=-\operatorname{diag}(\mathbf A_{\mathcal S}^{\mathsf T}\mathbf Q^2
\mathbf A_{\mathcal S})
+2\mathbf D_{\mathbf p}\mathbf A_{\mathcal S}^{\mathsf T}\mathbf Q\mathbf r.
\end{aligned}
\end{equation}
The second line has $i$th component
$-\mathbf a_i^{\mathsf T}\mathbf Q^2\mathbf a_i
+2p_i\mathbf a_i^{\mathsf T}\mathbf Q\mathbf r$.
The first term has $\ell_2$ norm $O_{\mathrm P}(\sqrt m)$, while the second is bounded
by
$2\|\mathbf D_{\mathbf p}\|_{\mathrm{op}}
\|\mathbf A\|_{\mathrm{op}}\|\mathbf Q\|_{\mathrm{op}}
\|\mathbf r\|=O_{\mathrm P}(\sqrt m)$. Hence
$\|\partial_\tau\overline{\boldsymbol\gamma}_{\mathcal S}\|_2=O_{\mathrm P}(\sqrt m)$.
Set $\mathbf r_\tau:=\partial_\tau\overline{\mathbf r}$. Then
$\mathbf r_\tau=-\mathbf Q\mathbf r-\mathbf Q\mathbf A_{\mathcal S}
\mathbf D_{\mathbf p}\partial_\tau\overline{\boldsymbol\gamma}_{\mathcal S}$,
so $\|\mathbf r_\tau\|=O_{\mathrm P}(\sqrt m)$. Mixed differentiation of the
profiled $\tau$-derivative gives
\begin{align}
    \nabla_{\mathbf y}\widehat\tau
    &=-\frac{2\mathbf r_\tau}{\partial S_\tau/\partial\tau},
    \label{eq:tau_gradient}\\
    \mathbf J_{\mathrm{joint}}
    &=\overline{\mathbf J}
      -\frac{2\mathbf r_\tau\mathbf r_\tau^{\mathsf T}}
      {\partial S_\tau/\partial\tau}.
    \label{eq:J_joint}
\end{align}
The same
Gaussian moment bounds used in Appendix~A give the moment bounds
\begin{equation}
\begin{aligned}
\mathbb E\left\|
  \frac{\partial\overline{\boldsymbol\gamma}_{\mathcal S}}{\partial\tau}
\right\|_2^2&=O(m),&
\mathbb E\|\mathbf r_\tau\|^2&=O(m),\\
\mathbb E\|\mathbf r_\tau\|^4&=O(m^2),&&
\end{aligned}
\label{eq:tau_response_moments}
\end{equation}
with the same orders uniformly over every fixed compact regular response
family used below. Indeed, the deterministic inverse-Hessian bound reduces
these quantities to fixed-degree Gaussian spectral-, column-, and residual-norm moments. Together with \eqref{eq:tau_curvature},
\eqref{eq:tau_gradient} gives
$\|\nabla_{\mathbf y}\widehat\tau\|_2=O_{\mathrm P}(m^{-1/2})$.
For a uniformly sampled $I$ (with $\mathbb E_I$ denoting conditional expectation over $I$),
\begin{equation}
\begin{aligned}
    \mathbb E_I[(\mathbf a_I^{\mathsf T}\mathbf r_\tau)^2
      \mid\mathbf A,\mathbf v]
    \leq\frac{\|\mathbf A\|_{\mathrm{op}}^2}{n}
      \|\mathbf r_\tau\|^2
    =O_{\mathrm P}(1).
\end{aligned}
\end{equation}
Thus $D_{\mathbf a_I}\widehat\tau=O_{\mathrm P}(m^{-1})$ for a uniformly sampled
direction. The rank-one correction \eqref{eq:J_joint} then gives
$m^{-1}\operatorname{tr}(\mathbf J_{\mathrm{joint}}-\overline{\mathbf J})
=O_{\mathrm P}(m^{-1})$.

We also use smooth fixed-$\tau$ sensitivities on the regular profiled branch:
$|\partial_\tau\eta_n|+|\partial_\tau V_n|+|\partial_\tau\overline h_n|=O_{\mathrm P}(1)$.
The first two follow from the resolvent identity and
$|\partial_\tau V_n|=2m^{-1}|\mathbf r^{\mathsf T}\mathbf r_\tau|$. For the third,
differentiating \eqref{eq:h_trace} gives
$\|\partial_\tau\mathbf H_{\mathcal S}\|_{\mathrm F}+\|\partial_\tau\mathbf L_{\mathcal S}\|_{\mathrm F}=O_{\mathrm P}(\sqrt m)$; with
$\|\mathbf H_{\mathcal S}^{-1}\|_{\mathrm{op}}=O(1)$ and
$\|\mathbf L_{\mathcal S}\|_{\mathrm{op}}=O_{\mathrm P}(1)$, the normalized trace derivative is $O_{\mathrm P}(1)$. For a uniformly sampled $I$,
$\partial_\tau\overline\mu_I=(\partial_\tau\overline\gamma_I)p_I+\overline\gamma_I\mathbf a_I^{\mathsf T}\mathbf r_\tau=O_{\mathrm P}(1)$, using
$n^{-1}\|\partial_\tau\overline{\boldsymbol\gamma}\|_2^2=O_{\mathrm P}(1)$ and the
average directional bound above.

\subsection{Direct Jointly Profiled Cavity Transport}

The profiled $\tau$ solution is data dependent, so the joint extension cannot be
obtained by inserting a random value of $\tau$ into a pointwise fixed-$\tau$
cavity theorem. We instead repeat the cavity transport on the jointly profiled
branch itself.

\begin{lemma}[Jointly profiled cavity transport]
\label{lem:joint_profiled_cavity}
Suppose Assumptions~\ref{ass:model}--\ref{ass:continuation} and
Assumption~\ref{ass:joint_regular} hold. Let $I$ be uniformly sampled. On the genuine
reduced joint branch driven by
$\mathbf u_I(t)=\mathbf y_{-I}+t\mathbf a_I$, let
$\widehat\tau_{I,-}(t)$, $\boldsymbol\gamma^{\mathrm{jt}}_{I,-}(t)$,
$\mathbf Q^{\mathrm{jt}}_{I,-}(t)$, $\mathbf r^{\mathrm{jt}}_{I,-}(t)$, and
$\mathbf J^{\mathrm{jt}}_{I,-}(t)$ denote the jointly selected value of $\tau$,
hyperparameters, resolvent, residual, and complete Jacobian; set
$V^{\mathrm{jt}}_{I,-}:=m^{-1}\|\mathbf r^{\mathrm{jt}}_{I,-}(0)\|^2$,
$h^{\mathrm{jt}}_{I,-}(t):=m^{-1}\operatorname{tr}\mathbf J^{\mathrm{jt}}_{I,-}(t)$ and
$\widehat\tau_{I,-}:=\widehat\tau_{I,-}(0)$. Then, for every fixed $B<\infty$,
\begin{align}
  &Z^{\mathrm{jt}}_{I,-}:=
  \mathbf a_I^{\mathsf T}\mathbf r^{\mathrm{jt}}_{I,-}(0)
  \ \big|\ \mathcal F_{-I}
  \sim \mathcal N(0,V^{\mathrm{jt}}_{I,-}),
  \label{eq:joint_fresh_field}\\
  &\int_{-B}^{B}
  \left|\frac{d}{dt}\widehat\tau_{I,-}(t)\right|dt
  =O_{\mathrm P}(m^{-1}),
  \label{eq:joint_response_tau_motion}\\
  &\sup_{|t|\le B}
  \|\mathbf A_{-I}^{\mathsf T}
      \mathbf Q^{\mathrm{jt}}_{I,-}(t)\mathbf a_I\|_\infty
  =O_{\mathrm P}(\rho_n).
  \label{eq:joint_response_cross}
\end{align}
Moreover, the jointly profiled response transport and the jointly profiled
one-coordinate insertion continuation change the normalized quantities
$\eta_n$, $V_n$, and $h_n$ only by $o_{\mathrm P}(1)$. At the pseudo-response endpoint
$e_I:=x_I-\mu_I^{\mathrm F}$ (with $\mathrm{ps}$ meaning $t=e_I$), the regular joint
correction continuation satisfies
\begin{align}
  \|\boldsymbol\gamma^{\mathrm{jt},\mathrm{ps}}_{-I}
    -\boldsymbol\gamma^{\mathrm F}_{-I}\|_2&=O_{\mathrm P}(\rho_n),
  \label{eq:joint_corr_gamma}\\
  |\widehat\tau^{\mathrm{jt},\mathrm{ps}}_{I}
    -\widehat\tau|&=O_{\mathrm P}(\rho_n/\sqrt m),
  \label{eq:joint_corr_tau}\\
  \left|\mathbf a_I^{\mathsf T}
  \mathbf r^{\mathrm{jt},\mathrm{ps}}_{I}-p_I^{\mathrm F}\right|&=O_{\mathrm P}(\rho_n).
  \label{eq:joint_corr_projection}
\end{align}
Consequently the scalar law and the two-coordinate empirical
scalar-law convergence of Theorem~\ref{thm:scalarization} hold on the jointly selected
branch, with the complete joint adaptive response in place of the fixed-$\tau$ adaptive response.
\end{lemma}

\begin{IEEEproof}
At $t=0$, interpret $\mathcal F_{-I}$ for the random index as
$\sigma(I,\mathbf A_{-I},\mathbf y_{-I},\zeta)$. The complete reduced
joint solution is measurable with respect to this sigma-field because both
$\boldsymbol\gamma$ and $\tau$ are selected from the reduced data and the
common auxiliary randomization. Hence $\mathbf a_I$ is
still conditionally independent, and \eqref{eq:joint_fresh_field} follows as in
Appendix~B.

For deterministic $t$, the jointly selected tuple based on
$\mathbf u_I(t)$ is measurable with respect to
$\mathcal G_I(t)=\sigma(I,\mathbf A_{-I},\mathbf u_I(t),\zeta)$. Let
$S_{\tau,I}(t)$ be its profiled $\tau$-derivative and $\mathbf r_{\tau,I}(t)$ the
fixed-$\tau$ profiled residual derivative at the selected value. By implicit differentiation of the profiled $\tau$ equation,
\begin{equation}
  \dot{\widehat\tau}_{I,-}(t)
  =-\frac{2\mathbf a_I^{\mathsf T}\mathbf r_{\tau,I}(t)}
  {\partial_\tau S_{\tau,I}(t)}.
  \label{eq:joint_response_taudot}
\end{equation}
Conditionally on $\mathcal G_I(t)$, the Gaussian regression formula of
Appendix~B applies to $\mathbf a_I$. Writing its conditional mean and
variance as $\mathbf m_I(t)$ and $\varsigma_I^2(t)\mathbf I_m$, respectively,
\begin{align}
  &\mathbb E[|\mathbf a_I^{\mathsf T}\mathbf r_{\tau,I}(t)|
      \mid\mathcal G_I(t)] \nonumber\\
  &\quad\le |\mathbf m_I(t)^{\mathsf T}\mathbf r_{\tau,I}(t)|
     +C\varsigma_I(t)\,\|\mathbf r_{\tau,I}(t)\|.
\end{align}
Here $\varsigma_I^2(t)=O(m^{-1})$ and
$\sup_{|t|\le B}\mathbb E\|\mathbf m_I(t)\|^2=O(m^{-1})$.
Together with
$\sup_{|t|\le B}\mathbb E\|\mathbf r_{\tau,I}(t)\|^2=O(m)$ from
\eqref{eq:tau_response_moments}, Cauchy--Schwarz gives
$\sup_{|t|\le B}\mathbb E|\mathbf a_I^{\mathsf T}\mathbf r_{\tau,I}(t)|=O(1)$.
Since $\partial_\tau S_{\tau,I}(t)\ge c_\tau m$, Fubini and Markov then give
\eqref{eq:joint_response_tau_motion}.

On a smooth segment the total hyperparameter response decomposes as
\begin{equation}
  \dot{\boldsymbol\gamma}_{\mathcal S}
  =2\mathbf H_{\mathcal S}^{-1}\mathbf D_{\mathbf p}
    \mathbf A_{\mathcal S}^{\mathsf T}\mathbf Q\mathbf a_I
   +\frac{\partial\overline{\boldsymbol\gamma}_{\mathcal S}}
   {\partial\tau}\,\dot{\widehat\tau}_{I,-}(t).
  \label{eq:joint_response_gammadot}
\end{equation}
The first term has $O_{\mathrm P}(1)$ norm and
$\|\partial_\tau\overline{\boldsymbol\gamma}_{\mathcal S}\|_2
=O_{\mathrm P}(\sqrt m)$. Hence its total variation over a bounded interval is
$O_{\mathrm P}(1)$. With $\mathbf c_I=\mathbf A_{-I}^{\mathsf T}\mathbf Q\mathbf a_I$ and
$\mathbf d_I^{(2)}=\mathbf A_{-I}^{\mathsf T}\mathbf Q^2\mathbf a_I$,
$\dot{\mathbf c}_I=-\mathbf G_{-I}(\dot{\boldsymbol\gamma}_{-I}\circ\mathbf c_I)-\dot{\widehat\tau}_{I,-}\mathbf d_I^{(2)}$.
The uniform resolvent/column bounds give $\|\mathbf d_I^{(2)}\|_\infty=O_{\mathrm P}(1)$;
starting from the initial $O_{\mathrm P}(\rho_n)$ Gaussian bound,
\eqref{eq:joint_response_tau_motion}--\eqref{eq:joint_response_gammadot} and
Gronwall give \eqref{eq:joint_response_cross}.

Let $\overline{\mathbf J}_{I,-}(t)$ freeze the profiled $\tau$ value and set
$\overline h_{I,-}(t):=m^{-1}\operatorname{tr}\overline{\mathbf J}_{I,-}(t)$, while
$\mathbf J^{\mathrm{jt}}_{I,-}(t)$ includes its response. By \eqref{eq:J_joint}, both are
$\mathcal G_I(t)$-measurable for deterministic $t$. The moment
bounds \eqref{eq:A_moment_bounds} and \eqref{eq:tau_response_moments}, together
with $\partial_\tau S_\tau\ge c_\tau m$, give uniformly for $|t|\le B$,
$\mathbb E\|\mathbf J^{\mathrm{jt}}_{I,-}(t)\|_{\mathrm{op}}^2=O(1)$ and
$\mathbb E\|\mathbf J^{\mathrm{jt}}_{I,-}(t)\|_{\mathrm F}^2=O(m)$. Therefore the same conditional quadratic-form argument as in Appendix~B gives
\begin{equation}
  \int_{-B}^{B}
  \left|\mathbf a_I^{\mathsf T}\mathbf J^{\mathrm{jt}}_{I,-}(t)
  \mathbf a_I-h^{\mathrm{jt}}_{I,-}(t)\right|dt
  \xrightarrow{\mathrm P}0.
  \label{eq:joint_response_quadint}
\end{equation}
Also, uniformly on the response path,
$m^{-1}\operatorname{tr}(\mathbf J^{\mathrm{jt}}_{I,-}-\overline{\mathbf J}_{I,-})
=O_{\mathrm P}(m^{-1})$. Along a smooth joint response segment,
$d\overline h_{I,-}/dt=\mathbf a_I^{\mathsf T}\nabla_{\mathbf u}\overline h_{I,-}
+(\partial_\tau\overline h_{I,-})\dot{\widehat\tau}_{I,-}$. The first term is $O_{\mathrm P}(m^{-1/2})$ by the fixed-$\tau$ calculation of
Appendix~B, now evaluated on the regular compact joint family, while
$|\partial_\tau\overline h_{I,-}|=O_{\mathrm P}(1)$. Equation
\eqref{eq:joint_response_tau_motion} makes the second term negligible after
integration. At a simple support transition, the fixed-$\tau$ part has the
$O_{\mathrm P}(m^{-1})$ normalized rank-one jump from Appendix~B; the noise-profile
correction has normalized trace $O_{\mathrm P}(m^{-1})$ on either side, so the complete
joint jump has the same order. The jointly profiled clause of
Assumption~\ref{ass:continuation} therefore controls their cumulative
contribution. Thus
\begin{equation}
  \sup_{|t|\le B}
  |h^{\mathrm{jt}}_{I,-}(t)-h^{\mathrm{jt}}_{I,-}(0)|=o_{\mathrm P}(1),
  \label{eq:joint_response_hstable}
\end{equation}
and the identical decomposition gives $o_{\mathrm P}(1)$ stability of $\eta$ and $V$.
Equations \eqref{eq:joint_response_quadint}--
\eqref{eq:joint_response_hstable} therefore integrate the joint response
exactly as in Appendix~B.

For completeness, consider next the jointly profiled insertion continuation
from the genuine joint cavity to the selected full point. At fixed $\tau$ the
nuisance KKT forcing is the one in Appendix~B and has $O_{\mathrm P}(1)$ Euclidean norm.
Let $S_{\tau,\lambda}$ denote the $\tau$-derivative after profiling the nuisance
hyperparameters while the explicit insertion parameter $\lambda$ is held
fixed, and let $\tau_\lambda$ be its selected $\tau$ value. The envelope theorem gives the fixed-$\tau$ profiled objective derivative
$\partial_\lambda\overline{\mathcal L}
=2x_ip_i+\gamma_i^{\mathrm F}g_i$, $g_i:=d_i-p_i^2$. Differentiating this identity with respect to
$\tau$ along the nuisance-profiled branch therefore gives the exact mixed
partial
\begin{equation}
  \left.\frac{d}{d\lambda}S_{\tau,\lambda}\right|_{\tau\,\mathrm{fixed}}
  =2x_i\mathbf a_i^{\mathsf T}\mathbf r_\tau
   +\gamma_i^{\mathrm F}\left(
      b_i+\boldsymbol\xi_i^{\mathsf T}
      \frac{\partial\overline{\boldsymbol\gamma}_{\mathcal S}}
           {\partial\tau}
    \right),
\end{equation}
where
$ b_i=-\mathbf a_i^{\mathsf T}\mathbf Q^2\mathbf a_i
      +2p_i\mathbf a_i^{\mathsf T}\mathbf Q\mathbf r$
and $\boldsymbol\xi_i=\partial_{\boldsymbol\gamma_{\mathcal S}}g_i$
is the cross-Hessian row between the externally inserted coordinate and the
nuisance active variables. The insertion cross-coupling bounds give
$\|\boldsymbol\xi_i\|_2=O_{\mathrm P}(1)$, while
$\|\partial_\tau\overline{\boldsymbol\gamma}_{\mathcal S}\|_2
=O_{\mathrm P}(\sqrt m)$. Hence the right-hand side is $O_{\mathrm P}(\sqrt m)$. Since the
jointly selected $\tau$ satisfies
$S_{\tau,\lambda}(\tau_\lambda)=0$ and its profiled curvature is at least
$c_\tau m$, implicit differentiation gives
$|\dot\tau_\lambda|=O_{\mathrm P}(m^{-1/2})$. Since
$\|\partial_\tau\overline{\boldsymbol\gamma}\|_2=O_{\mathrm P}(\sqrt m)$, the total
nuisance velocity remains $O_{\mathrm P}(1)$. The additional residual term
$\mathbf r_\tau\dot\tau_\lambda$ is also $O_{\mathrm P}(1)$ in norm. Consequently the
cross-coupling remains $O_{\mathrm P}(\rho_n)$, and the normalized $\eta$, $V$, and
nuisance-profiled $h$ changes are $o_{\mathrm P}(1)$ after the same smooth-segment and
transition-jump calculation as above. At $\lambda=1$ the state itself is the
selected full joint KKT point, but the insertion-path Jacobian still freezes
the $i$th hyperparameter. If $i$ is active, releasing that one scalar response
changes the fixed-$\tau$ residual Jacobian by the same rank-one Schur
correction as in Appendix~B, hence by $O_{\mathrm P}(m^{-1})$ after normalized trace; if
$i$ is strictly inactive the two fixed-$\tau$ Jacobians coincide locally, and
the same conclusion holds at a regular one-sided boundary. Profiling the
single $\tau$ response changes the normalized trace by another $O_{\mathrm P}(m^{-1})$
through \eqref{eq:J_joint}. Thus the insertion endpoint also matches the
complete full joint adaptive response up to $o_{\mathrm P}(1)$.

It remains to connect the joint pseudo-response endpoint to the full nuisance
state. We use the same explicit correction geometry as in Appendix~B,
\begin{align}
  \mathbf y_\theta^{\mathrm{c}}
  &=\mathbf u_i^{\mathrm{ps}}+\theta\mu_i^{\mathrm F}\mathbf a_i,\\
  \mathbf C_\theta^{\mathrm{c}}
  &=\tau^{\mathrm{c}}(\theta)\mathbf I_m
    +\mathbf A_{-i}\boldsymbol\Gamma_{-i}^{\mathrm{c}}(\theta)
       \mathbf A_{-i}^{\mathsf T}\\
  &\quad+\theta\gamma_i^{\mathrm F}\mathbf a_i\mathbf a_i^{\mathsf T},
  \qquad 0\le\theta\le1.
\end{align}
Set $\mathbf Q_\theta^{\mathrm c}:=(\mathbf C_\theta^{\mathrm c})^{-1}$,
$\mathbf r_\theta^{\mathrm c}:=\mathbf Q_\theta^{\mathrm c}\mathbf y_\theta^{\mathrm c}$, and
$p_i^{\mathrm c}:=\mathbf a_i^{\mathsf T}\mathbf r_\theta^{\mathrm c}$. Both nuisance
hyperparameters and the model-noise parameter are reoptimized. The path
is anchored at the full joint state at $\theta=1$ and at the genuine joint
pseudo-response state at $\theta=0$. Along this auxiliary correction path,
$S_\tau$ denotes the $\tau$-derivative after profiling the hyperparameters that are
free on the path (the nuisance coordinates); the curvature clause of
Assumption~\ref{ass:joint_regular} applies in this auxiliary profiled sense.
Define $\Delta_i:=p_i^{\mathrm F}-p_i^{\mathrm c}$,
$s_i:=\mathbf a_i^{\mathsf T}\mathbf Q_\theta^{\mathrm c}\mathbf a_i$,
$\mathbf c_i^{\mathrm c}:=\mathbf A_{\mathcal S}^{\mathsf T}\mathbf Q_\theta^{\mathrm c}\mathbf a_i$,
and $\mathbf b:=\partial_\tau\mathbf g_{\mathcal S}$,
$s_{\tau\tau}=\partial_\tau s_\tau$, and
$\kappa_\tau=s_{\tau\tau}-\mathbf b^{\mathsf T}
\mathbf H_{\mathcal S}^{-1}\mathbf b
=\partial_\tau S_\tau\ge c_\tau m$.
By equality of mixed partials on the current smooth chart,
$\partial_{\boldsymbol\gamma}s_\tau=\partial_\tau\mathbf g=\mathbf b$.
The regular bounds give $\|\mathbf b\|_2=O_{\mathrm P}(\sqrt m)$. Differentiating the
joint KKT system yields
\begin{align}
  \mathbf H_{\mathcal S}\dot{\boldsymbol\gamma}_{\mathcal S}
  +\mathbf b\dot\tau
  & =\mathbf f_\gamma,
  \label{eq:joint_corr_block1}\\
  \mathbf b^{\mathsf T}\dot{\boldsymbol\gamma}_{\mathcal S}
  +s_{\tau\tau}\dot\tau
  & =f_\tau,
  \label{eq:joint_corr_block2}
\end{align}
where
\begin{align}
  \mathbf f_\gamma
  &=\gamma_i^{\mathrm F}(\mathbf c_i^{\mathrm c})^{\circ2}
    +2\gamma_i^{\mathrm F}\Delta_i\mathbf D_{\mathbf p}\mathbf c_i^{\mathrm c},\\
  f_\tau
  &=\gamma_i^{\mathrm F}\mathbf a_i^{\mathsf T}\mathbf Q^2\mathbf a_i
    +2\gamma_i^{\mathrm F}\Delta_i\mathbf r^{\mathsf T}\mathbf Q\mathbf a_i.
\end{align}
Eliminating $\dot{\boldsymbol\gamma}$ gives
\begin{equation}
  \dot\tau
  =\frac{f_\tau-
  \mathbf b^{\mathsf T}\mathbf H_{\mathcal S}^{-1}\mathbf f_\gamma}
  {\kappa_\tau}.
  \label{eq:joint_corr_taudot}
\end{equation}
A coarse Gronwall argument gives bounded $\Delta_i$ and bounded total
nuisance velocity. The cross-coupling equation now contains only the extra
term $-\dot\tau\,\mathbf A_{-i}^{\mathsf T}\mathbf Q^2\mathbf a_i$; the
coarse bound $|\dot\tau|=O_{\mathrm P}(m^{-1/2})$ and the $O_{\mathrm P}(1)$ maximum norm of this
vector preserve
$\sup_\theta\|\mathbf c_i^{\mathrm c}(\theta)\|_\infty=O_{\mathrm P}(\rho_n)$ because
$m^{-1/2}=O(\rho_n)$. Hence
\begin{equation}
  \|\mathbf f_\gamma\|_2
  \le C\{\rho_n+|\Delta_i|\},
  \qquad
  |f_\tau|\le C\{1+\sqrt m|\Delta_i|\}.
\end{equation}
Equations \eqref{eq:joint_corr_block1}--\eqref{eq:joint_corr_taudot} give
$|\dot\tau|\le C(\rho_n+|\Delta_i|)/\sqrt m$ and
$\|\dot{\boldsymbol\gamma}_{\mathcal S}\|_2\le C(\rho_n+|\Delta_i|)$.
Finally,
\begin{equation}
  \dot p_i
  =\gamma_i^{\mathrm F}s_i\Delta_i
   -(\mathbf D_{\mathbf p}\mathbf c_i^{\mathrm c})^{\mathsf T}
     \dot{\boldsymbol\gamma}_{\mathcal S}
   -\dot\tau\,\mathbf a_i^{\mathsf T}\mathbf Q\mathbf r.
\end{equation}
Since $|\mathbf a_i^{\mathsf T}\mathbf Q\mathbf r|=O_{\mathrm P}(\sqrt m)$ by the
operator bounds, the preceding estimates imply
$|\dot\Delta_i|\le C|\Delta_i|+O_{\mathrm P}(\rho_n)$. As
$\Delta_i(1)=0$, Gronwall yields
$\sup_\theta|\Delta_i(\theta)|=O_{\mathrm P}(\rho_n)$, which sharpens the displayed
velocity bounds to
$|\dot\tau|=O_{\mathrm P}(\rho_n/\sqrt m)$ and
$\|\dot{\boldsymbol\gamma}\|_2=O_{\mathrm P}(\rho_n)$. Integration proves
\eqref{eq:joint_corr_gamma}--\eqref{eq:joint_corr_projection}.

Combining the integrated response identity with
\eqref{eq:joint_corr_projection} gives
\begin{equation}
  p_I=Z^{\mathrm{jt}}_{I,-}
      +h^{\mathrm{jt}}_{I,-}(0)(x_I-\mu_I)+o_{\mathrm P}(1),
\end{equation}
and the jointly profiled insertion bounds replace the cavity large-system state
quantities by their full-data counterparts. This proves the one-coordinate
joint scalar-law representation. For the two-coordinate argument, keep $I$ deleted and
insert $J$ only up to the selected one-coordinate \emph{joint} cavity based on
$(\mathbf A_{-I},\mathbf y_{-I})$, rather than toward the full-data joint
point. This nested reduced joint continuation, including its moving $\tau$ value,
is measurable with respect to
$\mathcal F_{-IJ}\vee\sigma(\mathbf a_J)$ and is therefore independent of
$\mathbf a_I$. Its residual velocity is $O_{\mathrm P}(1)$, including the term
$\mathbf r_\tau\dot\tau$, so its residual displacement is $O_{\mathrm P}(1)$ and
the projection onto the still conditionally independent $\mathbf a_I$ is $o_{\mathrm P}(1)$ in $L^2$.
Interchanging $I$ and $J$ gives the analogous statement. The double-cavity
conditional independence and the corrected $L^2$ bias--covariance argument of
Appendix~B therefore apply to the joint branch as well, proving the empirical
statement without introducing any full-data quantity into the conditioning
sigma-field.
\end{IEEEproof}

The adaptive leave-one-out diagonal resolvent relation of Appendix~C also extends directly
to this joint transport. Set
$s_{I,-}^{\mathrm{jt}}(t):=\mathbf a_I^{\mathsf T}\mathbf Q^{\mathrm{jt}}_{I,-}(t)\mathbf a_I$ and
$\eta_{I,-}^{\mathrm{jt}}(t):=m^{-1}\operatorname{tr}\mathbf Q^{\mathrm{jt}}_{I,-}(t)$. At $t=0$, conditional independence gives
$s_{I,-}^{\mathrm{jt}}(0)-\eta_{I,-}^{\mathrm{jt}}(0)=O_{\mathrm P}(m^{-1/2})$. Along a smooth joint response
segment, set $\boldsymbol\gamma_\tau:=\partial_\tau
\overline{\boldsymbol\gamma}_{-I}$, $q_j:=m^{-1}\mathbf a_j^{\mathsf T}\mathbf Q^2\mathbf a_j$,
and $\mathbf q:=(q_j)_{j\ne I}$. The chain rule
gives
\begin{align}
  &\frac{d}{dt}(s_{I,-}^{\mathrm{jt}}-\eta_{I,-}^{\mathrm{jt}})
  ={}\left.\frac{d}{dt}(s_{I,-}-\eta_{I,-})\right|_{\tau\ \mathrm{fixed}} -\dot{\widehat\tau}_{I,-}(t) \nonumber\\
  &\ \ 
  \times \Bigg\{\mathbf a_I^{\mathsf T}\mathbf Q^2\mathbf a_I
  -m^{-1}\operatorname{tr}\mathbf Q^2
  +\sum_{j\ne I}(\boldsymbol\gamma_\tau)_j
   \big(c_{jI}^2-q_j\big)\Bigg\}.
\end{align}
The first term is $O_{\mathrm P}(\rho_n)$ by Appendix~C and
\eqref{eq:joint_response_cross}. The explicit quadratic-form difference in
braces is $O_{\mathrm P}(1)$ under the uniform resolvent and column-norm envelope, while
\begin{align}
  \left|\sum_{j\ne I}(\boldsymbol\gamma_\tau)_j(c_{jI}^2-q_j)\right|
  &\le \|\boldsymbol\gamma_\tau\|_2
      \{\|\mathbf c_I^{\circ2}\|_2+\|\mathbf q\|_2\}\\
  &=O_{\mathrm P}(\sqrt{\log n}).
\end{align}
Equation \eqref{eq:joint_response_tau_motion} therefore makes the integrated
$\tau$-induced contribution
$O_{\mathrm P}(\sqrt{\log n}/m)=o_{\mathrm P}(1)$. Hence, evaluating at $t=e_I$ (superscript
$\mathrm{ps}$),
$s_I^{\mathrm{jt},\mathrm{ps}}-\eta_I^{\mathrm{jt},\mathrm{ps}}=o_{\mathrm P}(1)$. The resolvent identity,
\eqref{eq:joint_corr_gamma}, and \eqref{eq:joint_corr_tau} give
\begin{equation}
  \|\mathbf R_I^{\mathrm F}-\mathbf Q_I^{\mathrm{jt},\mathrm{ps}}\|_{\mathrm{op}}
  =O_{\mathrm P}(\rho_n),
\end{equation}
and the single covariance reinsertion changes the normalized trace by
$O_{\mathrm P}(m^{-1})$. Therefore
\begin{equation}
  s_I^{\mathrm F}-\eta_n=o_{\mathrm P}(1).
\end{equation}
On the joint branch as well. Consequently the scalar KKT relation
\eqref{eq:scalar_inactive}--\eqref{eq:scalar_active} and the active relation
\eqref{eq:scalar_kkt_graph} are unchanged under joint noise estimation.

The lemma also supplies the direct joint version of the adaptive-response
transport. For deterministic response points,
$\mathbf J^{\mathrm{jt}}_{I,-}(t)$ is measurable with respect to
$\mathcal G_I(t)$. We now verify the derivative envelope used by the
random-endpoint interpolation directly on the jointly profiled path.

On a smooth joint response segment, \eqref{eq:joint_response_gammadot} and the
pointwise consequence of \eqref{eq:joint_response_taudot},
\begin{equation}
  |\dot{\widehat\tau}_{I,-}(t)|
  \le
  \frac{2\|\mathbf a_I\|\,\|\mathbf r_{\tau,I}(t)\|}
       {c_\tau m}
  =O_{\mathrm P}(m^{-1/2}),
\end{equation}
which gives
\begin{equation}
  \|\dot{\boldsymbol\gamma}_{\mathcal S}\|_2=O_{\mathrm P}(1).
\end{equation}
Thus the total resolvent derivative is
\begin{equation}
  \dot{\mathbf Q}
  =-\mathbf Q
   \left\{
      \mathbf A_{\mathcal S}
      \operatorname{diag}(\dot{\boldsymbol\gamma}_{\mathcal S})
      \mathbf A_{\mathcal S}^{\mathsf T}
      +\dot{\widehat\tau}_{I,-}\mathbf I_m
   \right\}\mathbf Q.
\end{equation}
The diagonal structure and the bounds above imply
$\|\dot{\mathbf Q}\|_{\mathrm{op}}=O_{\mathrm P}(1)$ and
$\|\dot{\mathbf Q}\|_{\mathrm F}=O_{\mathrm P}(1)$ uniformly on every fixed regular response
interval. The complete residual derivative is
$\dot{\mathbf r}=\mathbf J^{\mathrm{jt}}\mathbf a_I=O_{\mathrm P}(1)$ in norm, because
both $\overline{\mathbf J}$ and the rank-one noise-profile correction in
\eqref{eq:J_joint} have $O_{\mathrm P}(1)$ operator norm. Consequently
$\|\dot{\mathbf p}_{\mathcal S}\|_2$,
$\|\dot{\mathbf G}_{\mathcal S}\|_{\mathrm F}$, and
$\|\dot{\mathbf H}_{\mathcal S}\|_{\mathrm F}$ are all $O_{\mathrm P}(1)$.

To control the derivative of the noise-profile correction, write
\begin{equation}
\begin{aligned}
  \mathbf T=\mathbf A_{\mathcal S}^{\mathsf T}\mathbf Q^2
             \mathbf A_{\mathcal S}, \quad \mathbf z=\mathbf A_{\mathcal S}^{\mathsf T}\mathbf Q\mathbf r,\\
  \mathbf b=-\operatorname{diag}(\mathbf T)+2\mathbf D_{\mathbf p}\mathbf z.
\end{aligned}
\end{equation}
Along the same total joint path,
\begin{equation}
  \dot{\mathbf T}
  =\mathbf A_{\mathcal S}^{\mathsf T}
    (\dot{\mathbf Q}\mathbf Q+\mathbf Q\dot{\mathbf Q})
    \mathbf A_{\mathcal S},
  \qquad
  \dot{\mathbf z}
  =\mathbf A_{\mathcal S}^{\mathsf T}
    (\dot{\mathbf Q}\mathbf r+\mathbf Q\dot{\mathbf r}).
\end{equation}
The first derivative has Frobenius norm $O_{\mathrm P}(1)$ while
$\|\dot{\mathbf z}\|_2=O_{\mathrm P}(\sqrt m)$, hence
$\|\dot{\mathbf b}\|_2=O_{\mathrm P}(\sqrt m)$.
Since
$\boldsymbol\gamma_\tau=-\mathbf H_{\mathcal S}^{-1}\mathbf b$,
\begin{equation}
  \dot{\boldsymbol\gamma}_\tau
  =\mathbf H_{\mathcal S}^{-1}\dot{\mathbf H}_{\mathcal S}
    \mathbf H_{\mathcal S}^{-1}\mathbf b
   -\mathbf H_{\mathcal S}^{-1}\dot{\mathbf b},
  \qquad
  \|\dot{\boldsymbol\gamma}_\tau\|_2=O_{\mathrm P}(\sqrt m).
\end{equation}
Termwise differentiation of
\begin{equation}
  \mathbf r_\tau
  =-\mathbf Q\mathbf r
   -\mathbf Q\mathbf A_{\mathcal S}\mathbf D_{\mathbf p}
    \boldsymbol\gamma_\tau,
\end{equation}
which gives
\begin{equation}
  \left\|\frac{d\mathbf r_{\tau,I}(t)}{dt}\right\|=O_{\mathrm P}(\sqrt m).
  \label{eq:joint_rtau_total}
\end{equation}

The profiled curvature is
\begin{equation}
  \kappa_\tau:=\partial_\tau S_\tau
  =s_{\tau\tau}
   -\mathbf b^{\mathsf T}\mathbf H_{\mathcal S}^{-1}\mathbf b,
  \quad
  s_{\tau\tau}
  =-\operatorname{tr}\mathbf Q^2+2\mathbf r^{\mathsf T}\mathbf Q\mathbf r.
\end{equation}
Using $\|\dot{\mathbf Q}\|_{\mathrm F}=O_{\mathrm P}(1)$,
$\|\dot{\mathbf r}\|=O_{\mathrm P}(1)$,
$\|\dot{\mathbf b}\|_2=O_{\mathrm P}(\sqrt m)$, and
$\|\dot{\mathbf H}_{\mathcal S}\|_{\mathrm{op}}=O_{\mathrm P}(1)$ gives
\begin{equation}
  \left|\frac{d\kappa_{\tau}(t)}{dt}\right|=O_{\mathrm P}(m).
  \label{eq:joint_kappa_total}
\end{equation}
(The term $\mathbf r^{\mathsf T}\dot{\mathbf Q}\mathbf r$ is bounded
coarsely by $O_{\mathrm P}(m)$; this order is already sufficient.)

Now write the noise-profile correction as
\begin{equation}
  \mathbf R_\tau^{\mathrm{prof}}:=\frac{2\mathbf r_\tau\mathbf r_\tau^{\mathsf T}}{\kappa_\tau}.
\end{equation}
Using $\|\mathbf r_\tau\|=O_{\mathrm P}(\sqrt m)$,
$\kappa_\tau\ge c_\tau m$, and
\eqref{eq:joint_rtau_total}--\eqref{eq:joint_kappa_total},
\begin{equation}
  \left\|\frac{d}{dt}\mathbf R_\tau^{\mathrm{prof}}\right\|_{\mathrm{op}}
  \le
  \frac{4\|\mathbf r_\tau\|\,\|\dot{\mathbf r}_\tau\|}{\kappa_\tau}
  +\frac{2\|\mathbf r_\tau\|^2|\dot\kappa_\tau|}{\kappa_\tau^2}
  =O_{\mathrm P}(1).
\end{equation}
The same total-path differentiation of the explicit formula
\eqref{eq:J_formula} in fact gives the sharper bound
$\|d\overline{\mathbf J}/dt\|_{\mathrm{op}}=O_{\mathrm P}(1)$. Indeed,
$\|\dot{\mathbf Q}\|_{\mathrm{op}}=O_{\mathrm P}(1)$,
$\|\dot{\mathbf p}_{\mathcal S}\|_2=O_{\mathrm P}(1)$, and
$\|\dot{\mathbf H}_{\mathcal S}\|_{\mathrm{op}}=O_{\mathrm P}(1)$, while all undifferentiated
factors in \eqref{eq:J_formula} have bounded operator norm. Therefore
\begin{equation}
  \sup_{|t|\le B,\ {\mathrm{smooth}}}
  \left\|\frac{d}{dt}\mathbf J^{\mathrm{jt}}_{I,-}(t)\right\|_{\mathrm{op}}
  =O_{\mathrm P}(1).
  \label{eq:joint_J_derivative_envelope}
\end{equation}
Repeating Appendix~D's deterministic-grid argument (conditioning separately
at each grid point) for the $\mathcal G_I(t)$-measurable matrices
$\mathbf J^{\mathrm{jt}}_{I,-}(t)$ and $\overline{\mathbf J}_{I,-}(t)$, using their common
moment envelopes and \eqref{eq:joint_J_derivative_envelope}, gives at $t=e_I$
\begin{align}
  \mathbf a_I^{\mathsf T}\mathbf J^{\mathrm{jt}}_{I,-}(e_I)\mathbf a_I
  -m^{-1}\operatorname{tr}\mathbf J^{\mathrm{jt}}_{I,-}(e_I)
  &=o_{\mathrm P}(1),\\
  \mathbf a_I^{\mathsf T}\overline{\mathbf J}_{I,-}(e_I)\mathbf a_I
  -m^{-1}\operatorname{tr}\overline{\mathbf J}_{I,-}(e_I)
  &=o_{\mathrm P}(1).
\end{align}
Since $m^{-1}\operatorname{tr}(\mathbf J^{\mathrm{jt}}-\overline{\mathbf J})=O_{\mathrm P}(m^{-1})$
uniformly, the two directional responses differ by $o_{\mathrm P}(1)$ at the endpoint. Hence the genuine fixed-$\tau$
nuisance-profiled response there equals $h_n^{\mathrm{joint}}+o_{\mathrm P}(1)$.

By Assumption~\ref{ass:joint_regular}, the typical-endpoint condition of
Assumption~\ref{ass:derivative} applies to this jointly profiled path.
On its no-crossing event, the algebraic full-nuisance and genuine joint
pseudo-response endpoints have the same nuisance active face, denoted
$\mathcal C$. Interpolate
linearly between their nuisance hyperparameters and noise variances while
keeping $\mathbf u_i^{\mathrm{ps}}$ fixed:
\begin{equation}
\begin{aligned}
\boldsymbol\gamma_{\mathcal C}(u)&=\boldsymbol\gamma_{\mathcal C}^{\mathrm{jt},\mathrm{ps}}
+u\,\Delta\boldsymbol\gamma_{\mathcal C},\\
\tau(u)&=\widehat\tau_i^{\mathrm{jt},\mathrm{ps}}+u\,\Delta\tau,\qquad u\in[0,1].
\end{aligned}
\end{equation}
where $\Delta\boldsymbol\gamma_{\mathcal C}:=
\boldsymbol\gamma_{\mathcal C}^{\mathrm F}-\boldsymbol\gamma_{\mathcal C}^{\mathrm{jt},\mathrm{ps}}$ and
$\Delta\tau:=\widehat\tau-\widehat\tau_i^{\mathrm{jt},\mathrm{ps}}$, with
$\|\Delta\boldsymbol\gamma_{\mathcal C}\|_2=O_{\mathrm P}(\rho_n)$ and
$|\Delta\tau|=O_{\mathrm P}(\rho_n/\sqrt m)$.
This is only a comparison path; with $\mathbf R(u)$ its resolvent, use the current-resolvent $s_i,\mathbf c_i,\mathbf G_i,\mathbf w_i,\mathbf K_i$ of Appendix~D. Direct differentiation gives
\begin{equation}
  \dot{\mathbf R}
  =-\mathbf R
   \left\{
      \mathbf A_{\mathcal C}
      \operatorname{diag}(\Delta\boldsymbol\gamma_{\mathcal C})
      \mathbf A_{\mathcal C}^{\mathsf T}
      +\Delta\tau\,\mathbf I_m
   \right\}\mathbf R.
\end{equation}
Starting from the genuine joint pseudo-response endpoint,
$\|\mathbf c_i\|_\infty=O_{\mathrm P}(\rho_n)$,
$\|\mathbf c_i\|_2=O_{\mathrm P}(1)$, and
$\|\mathbf p_{\mathcal C}\|_\infty=O_{\mathrm P}(1)$. Along the interpolation,
\begin{equation}
  \dot{\mathbf p}_{\mathcal C}
  =-\mathbf G_i
    (\Delta\boldsymbol\gamma_{\mathcal C}\circ\mathbf p_{\mathcal C})
   -\Delta\tau\,\mathbf A_{\mathcal C}^{\mathsf T}
     \mathbf R^2\mathbf u_i^{\mathrm{ps}}.
\end{equation}
Hence
$\|\dot{\mathbf p}_{\mathcal C}\|_\infty
\le C\rho_n\|\mathbf p_{\mathcal C}\|_\infty+C\rho_n$;
Gronwall preserves the $O_{\mathrm P}(1)$ maximum-norm bound over $u\in[0,1]$.
The same diagonal-perturbation bounds used in
\eqref{eq:D_reaction_transfer}, together with
$|\Delta\tau|=O_{\mathrm P}(\rho_n/\sqrt m)$, then give uniformly for $u\in[0,1]$,
$\|\dot{\mathbf p}_{\mathcal C}\|_2=O_{\mathrm P}(\rho_n)$,
$\|\dot{\mathbf c}_i\|_2=O_{\mathrm P}(\rho_n^2)$,
$\|\dot{\mathbf G}_i\|_{\mathrm F}=O_{\mathrm P}(\rho_n)$, and
$\|\dot{\mathbf K}_i\|_{\mathrm{op}}=O_{\mathrm P}(\rho_n)$. The extra noise-variance
terms have size $O_{\mathrm P}(|\Delta\tau|\sqrt m)=O_{\mathrm P}(\rho_n)$; the remaining
bounds are the same diagonal-perturbation estimates as in
\eqref{eq:D_reaction_transfer}.

Since $\mathbf K_i(0)\succeq c_H\mathbf I$ and
$\sup_u\|\mathbf K_i(u)-\mathbf K_i(0)\|_{\mathrm{op}}=O_{\mathrm P}(\rho_n)$, Weyl's inequality gives
$\inf_u\lambda_{\min}\mathbf K_i(u)\ge c_H/2$ with high probability.
For $\ell_i=s_i-2\mathbf w_i^{\mathsf T}\mathbf K_i^{-1}\mathbf w_i$,
the same bounds give $|\dot s_i|=O_{\mathrm P}(\rho_n^2)$,
$\|\dot{\mathbf w}_i\|_2=O_{\mathrm P}(\rho_n)$, and
$\sup_u|\dot\ell_i|=O_{\mathrm P}(\rho_n)$; hence the two endpoint responses differ
by $O_{\mathrm P}(\rho_n)=o_{\mathrm P}(1)$.

At the full point let $\boldsymbol\mu^{\mathrm{joint}}$ be the joint posterior mean. The coordinate and empirical corrections satisfy
\begin{align}
  \mathbf a_I^{\mathsf T}
  (\nabla_{\mathbf y}\mu_I^{\mathrm{joint}}-\nabla_{\mathbf y}\overline\mu_I)
  &=(\partial_\tau\overline\mu_I)D_{\mathbf a_I}\widehat\tau=O_{\mathrm P}(m^{-1}),\\
  \frac1n\sum_i(\partial_\tau\overline\mu_i)^2&=O_{\mathrm P}(1),\\
  \frac1n\sum_i(D_{\mathbf a_i}\widehat\tau)^2&=O_{\mathrm P}(m^{-2}),\\
  \frac1n\sum_i|(\partial_\tau\overline\mu_i)D_{\mathbf a_i}\widehat\tau|
  &=O_{\mathrm P}(m^{-1}).
\end{align}
The middle bounds follow from $\|\partial_\tau\overline{\boldsymbol\gamma}\|_2^2=O_{\mathrm P}(m)$,
$n^{-1}\|\mathbf A^{\mathsf T}\mathbf r_\tau\|^2=O_{\mathrm P}(1)$, and \eqref{eq:tau_gradient}; the last is Cauchy--Schwarz. Thus the two adaptive responses have the same limit. Differentiating the full joint identity
$\mathbf y=\mathbf A\boldsymbol\mu^{\mathrm{joint}}+\widehat\tau\mathbf r$ gives
\begin{equation}
  \mathbf I_m
  =\mathbf A D_{\mathbf y}\boldsymbol\mu^{\mathrm{joint}}
   +\widehat\tau\mathbf J_{\mathrm{joint}}
   +\mathbf r(\nabla_{\mathbf y}\widehat\tau)^{\mathsf T}.
\end{equation}
The extra normalized trace is
$m^{-1}O_{\mathrm P}(\|\mathbf r\|\|\nabla_{\mathbf y}\widehat\tau\|)
=O_{\mathrm P}(m^{-1})$. Hence the joint branch has the same adaptive-response limit; the noise balance is verified directly below.

\subsection{Full-to-Reduced $\tau$ Transfer}

Along the joint response path,
\eqref{eq:joint_response_tau_motion} and $e_I=O_{\mathrm P}(1)$ give
$|\widehat\tau_I^{\mathrm{jt},\mathrm{ps}}-\widehat\tau_{I,-}|=O_{\mathrm P}(m^{-1})$;
combining this with \eqref{eq:joint_corr_tau} yields
\begin{equation}
  |\widehat\tau-\widehat\tau_{I,-}|
  =O_{\mathrm P}\!\left(\frac{\rho_n}{\sqrt m}\right)
  =O_{\mathrm P}\!\left(\frac{\sqrt{\log n}}{m}\right)
  =o_{\mathrm P}(m^{-1/2}),
  \label{eq:tau_root_transfer_app}
\end{equation}
which is the rate used in Section~\ref{sec:noise}; large-system state replacement follows from Lemma~\ref{lem:joint_profiled_cavity}.

By Assumption~\ref{ass:joint_regular}, the joint selected residual map inherits the
locally Lipschitz, piecewise-$C^1$ chart structure of Assumption~\ref{ass:regular};
positive profiled curvature and the implicit-function formula identify its weak Jacobian with
$\mathbf J_{\mathrm{joint}}$ almost everywhere. Thus \eqref{eq:J_joint},
\eqref{eq:tau_response_moments}, and Stein's identity \eqref{eq:C_second_order_stein} give
\begin{align}
  \mathbb E\|\mathbf J_{\mathrm{joint}}\|_{\mathrm F}^2
  &\le C\mathbb E\|\overline{\mathbf J}\|_{\mathrm F}^2
    +\frac{C}{m^2}\mathbb E\|\mathbf r_\tau\|^4=O(m),\\
  m^{-1}\mathbf v^{\mathsf T}\mathbf r-\sigma_0^2 h_n^{\mathrm{joint}}
  &\xrightarrow{\mathrm P}0.
\end{align}

\subsection{Branchwise Reconstruction Error}

By \eqref{eq:A_basic_bounds}, $|\mu_i|\le\gamma_{\max}\tau_\star^{-1/2}\|\mathbf a_i\|$;
bounded $x_i$ and Gaussian column moments give uniform integrability, hence the
MSE limit stated in Section~\ref{sec:noise} (in probability for deterministic limits).

\subsection{Closed MSE Relation}

Differentiating the scalar KKT relation gives
$\partial_U\widehat X=(2+\eta\Gamma)/[(2h-\eta)+\eta h\Gamma]$ on the active
branch and zero otherwise; moreover $P=-h(\widehat X-X)+\sqrt V\,G$.
For each realized branch limit $(\mathcal M,\nu_b)$, apply scalar Stein conditionally on that limit; together with Theorem~\ref{thm:susceptibility} and \eqref{eq:balance_noise}, this yields
\begin{align}
1-\tau h
&=\frac{h}{\delta}
  \left\langle\frac{\partial\widehat X}{\partial U}\right\rangle_b,
\label{eq:E_susceptibility_scalar}\\
\langle(\widehat X-X)G\rangle_b
&=\sqrt V\left\langle\frac{\partial\widehat X}{\partial U}\right\rangle_b,
\label{eq:E_stein}\\
\langle(\widehat X-X)P\rangle_b
&=-h\mathcal E_b
  +V\left\langle\frac{\partial\widehat X}{\partial U}\right\rangle_b,\\
\sigma_0^2 h
&=\frac{V}{h}-\frac{h}{\delta}\mathcal E_b.
\end{align}
Since $h\neq0$ by Theorem~\ref{thm:susceptibility}, the last identity gives
\eqref{eq:closed_mse}; with $V=\eta$ it gives
\eqref{eq:closed_mse_learned}--\eqref{eq:alpha_mse}.

%% file: sbl_tsp_technical_main.bbl
\begin{thebibliography}{40}
\bibitem{Tipping2001}
M.~E. Tipping, ``Sparse Bayesian learning and the relevance vector machine,''
\emph{J. Mach. Learn. Res.}, vol.~1, pp.~211--244, 2001.
\bibitem{TippingFaul2003}
M.~E. Tipping and A.~C. Faul, ``Fast marginal likelihood maximisation for sparse Bayesian models,'' in \emph{Proc. Int. Workshop Artif. Intell. Statist. (AISTATS)}, Key West, FL, USA, 2003, pp.~276--283.
\bibitem{WipfRao2004}
D.~P. Wipf and B.~D. Rao, ``Sparse Bayesian learning for basis selection,'' \emph{IEEE Trans. Signal Process.}, vol.~52, no.~8, pp.~2153--2164, Aug.~2004.
\bibitem{Dai2015}
J.~Dai, N.~Hu, W.~Xu, and C.~Chang, ``Sparse Bayesian learning for DOA estimation with mutual coupling,'' \emph{Sensors}, vol.~15, no.~10, pp.~26267--26280, Oct.~2015, doi: 10.3390/s151026267.
\bibitem{Srivastava2019}
S.~Srivastava, A.~Mishra, A.~Rajoriya, A.~K. Jagannatham, and G.~Ascheid, ``Quasi-static and time-selective channel estimation for block-sparse millimeter wave hybrid MIMO systems: Sparse Bayesian learning (SBL) based approaches,'' \emph{IEEE Trans. Signal Process.}, vol.~67, no.~5, pp.~1251--1266, Mar.~2019, doi: 10.1109/TSP.2018.2890058.
\bibitem{Ojeda2018}
A.~Ojeda, K.~Kreutz-Delgado, and T.~Mullen, ``Fast and robust block-sparse Bayesian learning for EEG source imaging,'' \emph{NeuroImage}, vol.~174, pp.~449--462, 2018.
\bibitem{Grebien2024}
S.~Grebien, E.~Leitinger, K.~Witrisal, and B.~H. Fleury, ``Super-resolution estimation of UWB channels including the dense component---An SBL-inspired approach,'' \emph{IEEE Trans. Wireless Commun.}, vol.~23, no.~8, pp.~10301--10318, Aug.~2024, doi: 10.1109/TWC.2024.3371352.
\bibitem{FaulTipping2001}
A.~C. Faul and M.~E. Tipping, ``Analysis of sparse Bayesian learning,'' in \emph{Adv. Neural Inf. Process. Syst.}, vol.~14, pp.~383--389, 2001.
\bibitem{WipfNagarajan2007}
D.~P. Wipf and S.~S. Nagarajan, ``A new view of automatic relevance determination,'' in \emph{Adv. Neural Inf. Process. Syst.}, vol.~20, pp.~1625--1632, 2007.
\bibitem{WipfNagarajan2010}
D.~P. Wipf and S.~S. Nagarajan, ``Iterative reweighted $\ell_1$ and $\ell_2$ methods for finding sparse solutions,''
\emph{IEEE J. Sel. Topics Signal Process.}, vol.~4, no.~2, pp.~317--329, Apr.~2010,
doi: 10.1109/JSTSP.2010.2042413.
\bibitem{WipfRaoNagarajan2011}
D.~P. Wipf, B.~D. Rao, and S.~S. Nagarajan, ``Latent variable Bayesian models for promoting sparsity,''
\emph{IEEE Trans. Inf. Theory}, vol.~57, no.~9, pp.~6236--6255, Sep.~2011,
doi: 10.1109/TIT.2011.2162174.
\bibitem{Hashemi2021}
A.~Hashemi, C.~Cai, G.~Kutyniok, K.-R.~M{\"u}ller, S.~S. Nagarajan, and S.~Haufe, ``Unification of sparse Bayesian learning algorithms for electromagnetic brain imaging with the majorization-minimization framework,'' \emph{NeuroImage}, vol.~239, p.~118309, Oct.~2021, doi: 10.1016/j.neuroimage.2021.118309.
\bibitem{Yu2024}
F.~Yu, L.~Shen, and G.~Song, ``Hyperparameter estimation for sparse Bayesian learning models,'' \emph{SIAM/ASA J. Uncertain. Quantif.}, vol.~12, no.~3, pp.~759--787, 2024, doi: 10.1137/24M162844X.
\bibitem{Song2024}
Y.~Song, Z.~Gong, Y.~Chen, and C.~Li, ``Towards inversion-free sparse Bayesian learning: A universal approach,'' \emph{IEEE Trans. Signal Process.}, vol.~72, pp.~4992--5006, 2024, doi: 10.1109/TSP.2024.3484908.
\bibitem{Ji2008}
S.~Ji, Y.~Xue, and L.~Carin, ``Bayesian compressive sensing,''
\emph{IEEE Trans. Signal Process.}, vol.~56, no.~6, pp.~2346--2356, Jun.~2008,
doi: 10.1109/TSP.2007.914345.
\bibitem{WipfRao2007}
D.~P. Wipf and B.~D. Rao, ``An empirical Bayesian strategy for solving the simultaneous sparse approximation problem,''
\emph{IEEE Trans. Signal Process.}, vol.~55, no.~7, pp.~3704--3716, Jul.~2007,
doi: 10.1109/TSP.2007.894265.
\bibitem{ZhangRao2013}
Z.~Zhang and B.~D. Rao, ``Extension of SBL algorithms for the recovery of block sparse signals with intra-block correlation,''
\emph{IEEE Trans. Signal Process.}, vol.~61, no.~8, pp.~2009--2015, Apr.~2013,
doi: 10.1109/TSP.2013.2241055.
\bibitem{PrasadMurthy2013}
R.~Prasad and C.~R. Murthy, ``Cram\'er--Rao-type bounds for sparse Bayesian learning,'' \emph{IEEE Trans. Signal Process.}, vol.~61, no.~3, pp.~622--632, 2013.
\bibitem{KhannaMurthy2022}
S.~Khanna and C.~R. Murthy, ``On the support recovery of jointly sparse Gaussian sources via sparse Bayesian learning,'' \emph{IEEE Trans. Inf. Theory}, vol.~68, no.~11, pp.~7361--7378, 2022.
\bibitem{Moderl2025}
J.~M{\"o}derl, E.~Leitinger, and B.~H. Fleury, ``General pruning criteria for fast SBL,'' \emph{IEEE Signal Process. Lett.}, vol.~32, pp.~4374--4378, 2025, doi: 10.1109/LSP.2025.3632230.
\bibitem{YangXieZhang2013}
Z.~Yang, L.~Xie, and C.~Zhang, ``Off-grid direction of arrival estimation using sparse Bayesian inference,''
\emph{IEEE Trans. Signal Process.}, vol.~61, no.~1, pp.~38--43, Jan.~2013,
doi: 10.1109/TSP.2012.2222378.
\bibitem{Gerstoft2016}
P.~Gerstoft, C.~F. Mecklenbr\"auker, A.~Xenaki, and S.~Nannuru, ``Multisnapshot sparse Bayesian learning for DOA,''
\emph{IEEE Signal Process. Lett.}, vol.~23, no.~10, pp.~1469--1473, Oct.~2016,
doi: 10.1109/LSP.2016.2598550.
\bibitem{YoshidaWatanabe2026}
T.~Yoshida and K.~Watanabe, ``Empirical Bayes estimation for lasso-type regularizers and its consistency,'' \emph{IEICE Trans. Fundamentals}, vol.~E109-A, no.~3, pp.~490--499, Mar.~2026, doi: 10.1587/transfun.2025TAP0008.
\bibitem{XiaoSlock2026ICASSP}
F.~Xiao and D.~T.~M. Slock, ``Large-system fixed-point law and deterministic closure for sparse Bayesian learning,''
in \emph{Proc. IEEE Int. Conf. Acoust., Speech Signal Process. (ICASSP)}, Barcelona, Spain, May 2026,
doi: 10.1109/ICASSP55912.2026.11464948.
\bibitem{Donoho2009}
D.~L. Donoho, A.~Maleki, and A.~Montanari, ``Message passing algorithms for compressed sensing,'' \emph{Proc. Natl. Acad. Sci. USA}, vol.~106, no.~45, pp.~18914--18919, 2009.
\bibitem{BayatiMontanari2011}
M.~Bayati and A.~Montanari, ``The dynamics of message passing on dense graphs, with applications to compressed sensing,'' \emph{IEEE Trans. Inf. Theory}, vol.~57, no.~2, pp.~764--785, 2011.
\bibitem{Rangan2019}
S.~Rangan, P.~Schniter, and A.~K. Fletcher, ``Vector approximate message passing,'' \emph{IEEE Trans. Inf. Theory}, vol.~65, no.~10, pp.~6664--6684, 2019.
\bibitem{VilaSchniter2013}
J.~P. Vila and P.~Schniter, ``Expectation-maximization Gaussian-mixture approximate message passing,'' \emph{IEEE Trans. Signal Process.}, vol.~61, no.~19, pp.~4658--4672, 2013.
\bibitem{BaoHanXu2025}
Z.~Bao, Q.~Han, and X.~Xu, ``A leave-one-out approach to approximate message passing,'' \emph{Ann. Appl. Probab.}, vol.~35, no.~4, pp.~2716--2766, Aug.~2025, doi: 10.1214/25-AAP2186.
\bibitem{BayatiMontanari2012}
M.~Bayati and A.~Montanari, ``The LASSO risk for Gaussian matrices,'' \emph{IEEE Trans. Inf. Theory}, vol.~58, no.~4, pp.~1997--2017, Apr.~2012, doi: 10.1109/TIT.2011.2174612.
\bibitem{Thrampoulidis2015}
C.~Thrampoulidis, S.~Oymak, and B.~Hassibi, ``Regularized linear regression: A precise analysis of the estimation error,'' in \emph{Proc. Conf. Learn. Theory (COLT)}, 2015, pp.~1683--1709.
\bibitem{Thrampoulidis2018}
C.~Thrampoulidis, E.~Abbasi, and B.~Hassibi, ``Precise error analysis of regularized $M$-estimators in high dimensions,'' \emph{IEEE Trans. Inf. Theory}, vol.~64, no.~8, pp.~5592--5628, Aug.~2018, doi: 10.1109/TIT.2018.2840720.
\bibitem{Stein1981}
C.~M. Stein, ``Estimation of the mean of a multivariate normal distribution,'' \emph{Ann. Statist.}, vol.~9, no.~6, pp.~1135--1151, 1981.
\bibitem{BellecZhang2021}
P.~C. Bellec and C.-H. Zhang, ``Second-order Stein: SURE for SURE and other applications in high-dimensional inference,'' \emph{Ann. Statist.}, vol.~49, no.~4, pp.~1864--1903, Aug.~2021, doi: 10.1214/20-AOS2005.
\bibitem{Vershynin2018}
R.~Vershynin, \emph{High-Dimensional Probability: An Introduction with Applications in Data Science}. Cambridge, U.K.: Cambridge Univ. Press, 2018.
\bibitem{RudelsonVershynin2013}
M.~Rudelson and R.~Vershynin, ``Hanson-Wright inequality and sub-Gaussian concentration,'' \emph{Electron. Commun. Probab.}, vol.~18, no.~82, pp.~1--9, 2013.
\end{thebibliography}
